\documentclass[aps,pra,showpacs,twoside,10pt,floatfix,twocolumn]{revtex4-2}

\usepackage[colorlinks=true, citecolor=blue, urlcolor=blue]{hyperref}
\usepackage{epsfig,newlfont,amssymb,amsfonts,amsmath,bm,subfigure,palatino,mathtools,amsthm,braket,soul,enumitem,color,graphics,graphicx,times,physics,makecell,multirow,float}
\usepackage[normalem]{ulem}

\newtheorem{proposition}{Proposition}

\begin{document}
\title{ Slow decay rate of correlations induced by long-range extended Dzyaloshinskii-Moriya interactions}

\author{Tanoy Kanti Konar$^{1}$, Leela Ganesh Chandra Lakkaraju$^{1,2,3}$, Aditi Sen(De)$^{1}$}
\affiliation{$^1$ Harish-Chandra Research Institute, A CI of Homi Bhabha National Institute, Chhatnag Road, Jhunsi, Allahabad - 211019, India}
\affiliation{$^2$ Pitaevskii BEC Center and Department of Physics, University of Trento, Via Sommarive 14, I-38123 Trento, Italy }
\affiliation{{$^3$} INFN-TIFPA, Trento Institute for Fundamental Physics and Applications, Trento, Italy}

\begin{abstract}
    We examine the impact of long-range  Dzyaloshinskii-Moriya (DM) interaction in the extended $XY$ model on the phase diagram as well as the static and dynamical properties of quantum and classical correlation functions. It is known that in the nearest-neighbor $XY$ model with DM interaction, the transition from the gapless chiral phase to a gapped one occurs when the strengths of the DM interaction and anisotropy coincide. We exhibit that the critical line gets modified with the range of interactions which decay according to power-law.   Specifically, instead of being gapless in the presence of a strong DM interaction,  a gapped region emerges which grows with the increase of the moderate fall-off rate (quasi-long range regime) in the presence of a transverse magnetic field. The gapless chiral phase can also be separated from a gapped one by the decay patterns of quantum mutual information and classical correlation with distant sites of the ground state which are independent of the fall-off rate in the gapless zone. We observe that the corresponding critical lines that depend on the fall-off rate can also be determined from the effective central charge involved in the scaling of entanglement entropy.  We illustrate that in a non-equilibrium setting, the relaxation dynamics of classical correlation,  the decay rate of total correlation, and the growth rate of entanglement entropy can be employed to uncover whether the evolving Hamiltonian and the Hamiltonian corresponding to the initial state are gapped or gapless. 
\end{abstract}
\maketitle

% \section{Things to Do}
% \begin{enumerate}
%     \item [1.] $C_{xx}$ to establish para and ferro phases. 
%     \item [2.] Chiral order vs R in static. 
%     \item [3.] h quench 
%     \item [4.] Entropy and bipartite entanglement

% \end{enumerate}
\section{Introduction}

Long-range order across quantum systems is desirable for a variety of quantum technological tasks, as it facilitates the sharing of correlation among distant parts of the system. Historically, such long-range order has been observed in many-body systems at quantum criticality or when the spectrum becomes gapless. In the case of the nearest-neighbor (NN) $XY$ model with nonvanishing anisotropy parameter, it has been demonstrated that the classical correlations decay with the distance between the spins  exponentially (polynomially) when it is away from (at) the criticality.
Beyond NN models, long-range interacting (LR) systems with power-law decay in the range of interactions have been shown to exhibit polynomial decay of correlations, even when the spectrum is not gapless. These systems become an intense topic of research in the last few years, especially in the context of building analog quantum simulator\cite{diessel_prr_2023, Defenu2024Jun, defenu_prl_2018, defenu_prb_2019, solfanelli_arxiv_2024, defenu_arxiv_2024} and due to their natural occurrence in several physical platforms such as the Rydberg atom arrays \cite{rydberg_review_experiments}, dipolar systems \cite{dipolar_longrange}, polar molecules \cite{cold_gas_long_range_review}, trapped-ion setups \cite{trapped_ion_1, trapped_ion_2, trapped_ion_3}, and cold atoms in cavities \cite{cold_atom_cavity_long_range_1, cold_atom_cavity_long_range_2}. Moreover, several established results that hold for nearest-neighbor systems came into the light of investigation in these LR systems in terms of Lieb-Robinson bound \cite{PhysRevB.93.125128}, area law \cite{area_law_longrange}, novel phases of matter  \cite{phases_longrange},  dynamical phases \cite{uhrich_prb_2020} and from the perspectives of usefulness in quantum technologies like  quantum metrology \cite{monika2023bettersensingvariablerangeinteractions} and quantum computation \cite{ghosh2023entanglementweightedgraphsuncovers}. 
%in the LR systems. 
%Also, dynamical phases observed in long-range model \cite{uhrich_prb_2020}. 
%From a technological perspective, LRI systems are also useful as they host highly entangled states, in terms of quantum-metrology \cite{monika2023bettersensingvariablerangeinteractions}.
%When the anisotropy parameter vanishes, the system becomes gapless for various of parameters.  
%A gapless spectrum is not necessary to establish long-range order in terms of PDC, 

%In addition, when the NN Dzyaloshinskii-Moriya (DM) interaction \cite{moriya_1960, moriya_prl_1960} is stronger than the anisotropy parameter  in the Hamiltonian, the spectrum becomes gapless and chiral which leads to the long-range order in terms of the polynomial decay of correlations \cite{}.

The nearest-neighbor Dzyaloshinskii-Moriya (DM) interaction, on the other hand,  involves an asymmetric exchange of spins caused by spin-orbit coupling \cite{DZYALOSHINSKY1958241, moriya_1960, moriya_prl_1960}.
%. This interaction  
It has been widely explored in a variety of solid-state compounds, including \( \text{Cu(C}_6\text{D}_5\text{COO})_2. 3\text{D}_2\text{O} \) \cite{comp_1, comp_2}, \( \text{Yb}_4\text{As}_3 \)\cite{comp_3,comp_4}, \(\text{Crl}_3\) \cite{jafari_experimental} etc, exhibiting fascinating magnetic properties, like gapless chiral phase when the DM interaction strength is stronger than the anisotropy parameter in the transverse Ising \cite{jafri_prb_2008, GR_phase_DM},  $XY$ \cite{jafari_XY_entanglement, roy_prb_2019, Zhong_2013} and the gamma model \cite{quantum_phase_DM_2}. It was also shown that the DM interaction can host a ground state in the gapless phase that contains logarithmic entanglement behavior, similar to the LR interacting systems which can enhance the fidelity in quantum teleportation \cite{DM_teleportation}, thermal entanglement \cite{roy_prb_2019, thermal_entanglement, jafari_thermal_gamma, jafari_thermal_heisenberg} and performance in quantum engines \cite{heat_engine_DM_1, heat_engine_DM_2}.  
%and in terms of Gr\"uneisen ratio to detect phase transition.
%influenced by quantum fluctuations under an applied magnetic field. 
%The study of entanglement properties of spin chains with DM interaction \cite{Zhong_2013} is established such as quantum teleportation, induced thermal entanglement, entanglement engines, and several more. From a fundamental perspective, research on quantum phase transitions is one of the cornerstones of the many-body physics community, in this regard, DM interaction has played a definitive role in this picture \cite{quantum_phase_DM_1, quantum_phase_DM_2} and in terms of Gr\"uneisen ratio to detect phase transition \cite{GR_phase_DM}. 
In the non-equilibrium domain, the DM interaction has also been examined in the context of non-equilibrium thermodynamics \cite{nonequi_thermo_DM}, dynamical quantum phase transitions \cite{nonequi_DM}, topological phase transition \cite{jafari_topolpgy_DM} and quantum speed limit \cite{qsl_DM}. 
%The interaction plays an interesting role in spin systems due to the induction of chiral and gapless phases, often occurring concurrently. It is well known that the gapless phase hosts a ground state that contains logarithmic entanglement behavior. A similar behavior was found in long-range systems. 

 In this paper, we explore the long-range extended $XY$ model in the presence of long-range extended DM interactions, which decay according to a power-law.  Specifically, we aim to analyze the responses of both long-range and DM interactions in equilibrium and non-equilibrium physics. In the former scenario, we present the phase-classification utilizing conventional order parameters, the law for the decreasing slopes of classical and quantum correlations, and the scaling of entanglement entropy (EE) with the corresponding central charge. In contrast to the NN model \cite{roy_prb_2019}, we show that given a fixed fall-off rates, there exists a gapped region in which the DM interaction is stronger than the anisotropy parameter. However, we establish that whenever the system is gapless, it is chiral. 
 Notably, in the static case, when the ground state belongs to the chiral gapless region, the classical correlation and quantum mutual information decay with constant decreasing exponent while this is not the case in the gapped regime.   In this model, we observe that the effective central charge associated with EE exhibits interesting behaviors: in the presence of LR and DM interactions, the central charge acts nonlinearly with multiple kinks that are absent from the extended Ising model without DM interactions.

  In comparison to the equilibrium scenario, the dynamical phases could not be efficiently characterized as they depend on the  initial state, evolving operator,  time and system parameters. 
   In order to distinguish between the gapless and gapped phases, we resort to a recently proposed relaxation dynamics of two-point classical correlations \cite{dyn_relax_floquet_2016, Nandy_2018, dyn_relax_2020, dyn_relax_2022} and entanglement entropy of the evolved state in the transient regimes, as well as mutual information and classical correlation in the steady state domains
  (cf. \cite{oscillation_dyn_phases_2017, DELFINO2022115643, ent_oscill_2020, heyl_2013, Heyl_2018, makki_prb_2022, Sirker_XXZ_2023, lakkaraju2023frameworkdynamicaltransitionslongrange} for other dynamical quantities). In our case, the initial state is  prepared by tuning the parameters of the Hamiltonian which have either gapless or gapped energy spectra while  evolution happens when the range of interactions is suddenly quenched, resulting in a gapless or gapped Hamiltonian. Specifically, in the transient  regime,  we show that the long-range DM interactions modify the exponent in the relaxation time of the two-point classical correlation between modes  and establish a generic  scaling law for the growth of entanglement entropy over time. Further, we  observe that  the behavior of mutual information and classical correlation between two arbitrary sites in the steady state differs depending on whether the Hamiltonians corresponding to the initial and dynamical states are in the same (gapped or gapless) phase or in different phases.

The paper is organized as follows. In Sec. \ref{sec:model}, we introduce the long-range extended $XY$ model with an additional long-range DM term and find the energy spectra and eigenvectors. While 
%in Sec. \ref{sec:phase_classification}, 
we identify the conventional phases in Sec. \ref{sec:gapless_chiral}, we indicate the decay of classical and quantum correlations in Sec. \ref{sec:static_corre_decay}. The behavior of  entanglement entropy and the corresponding central charge with respect to system parameters are studied in Sec. \ref{sec:entropy_charge}. We then move on to the investigation of features in dynamical states, namely the relaxation dynamics of classical correlations, and the rate of entropy growth in the transient regime in Secs. \ref{sec:dynamics} and  \ref{sec:dynamics_entropy} respectively while we examine the decay of mutual information with distance between the spins  in the steady state  in Sec. \ref{sec:steady_state_corre}. The results are summarized  in Sec. \ref{sec:conclusion}.

\section{Emergent gapped phase with strong DM interactions: Phase diagram}
\label{sec:model}

The  critical lines  for the long-range interacting $XY$ spin models in the presence of Coloumb-like interaction, described by a parameter \(\alpha\) and for the nearest-neighbor transverse $XY$ model with DM interaction  are known in the literature \cite{roy_prb_2019}. However,  the effects of long-range DM interactions has never been addressed before. In this situation, a possible query can be  -- can the introduction of long-range DM interactions change the critical lines and phases of the LR Ising models? We answer this question affirmatively by first presenting the analytical prescription of this model for obtaining energy spectra and the eigenvectors, thereby leading to the modified phase-portrait of the extended long-range $XY$ model \cite{sadhukhan_prb_2020, sinha_prb_2020} having long-range extended Dzyaloshinskii-Moriya interaction both in the thermodynamic limit and for finite systems.  In particular, we report the gapless chiral phase and its dependence on the fall-off rate, \(\alpha\). In addition, we identify appropriate order parameters that can reveal \(\alpha\)-dependence in the modified magentic phase diagram. From the decay patterns of classical correlation including Landau parameters and quantum correlations in terms of quantum mutual information and block entanglement entropy, we again separate the gapless region from the gapped one in this model. 

\subsection{Energy spectrum of the long-range XY and DM interactions}

The Hamiltonian, describing long-range extended XY model having long-range extended DM interaction in the presence of transverse magnetic field, reads as
%\section{Long range DM interaction: model and dispersion relation}% \textcolor{red}{Jafari paper - solid state motivation for DM interaction. }%We first present the description of the extended long-range \(XY\) model \cite{sadhukhan_prb_2020, sinha_prb_2020} which includes a long-range Dzyaloshinskii-Moriya (DM) interaction. Even in presence of long-range interactions, can be solved analytically in the thermodynamic limit and for a finite-size system. It reads as for an even number of sites, \(N\), as
\begin{eqnarray}
    \nonumber H&=&\sum_{j=1}^{N} \sum_{r=1}^{\frac{N}{2}} -\frac{J_r^\prime}{\mathcal{N}} \Bigg [\frac{1+\gamma}{4}\sigma_j^x\mathbb{Z}_r^z\sigma_{j+r}^x\nonumber+\frac{1-\gamma}{4}\sigma_j^y\mathbb{Z}_r^z\sigma_{j+r}^y\nonumber\\&&+\frac{D'}{4}(\sigma_j^x\mathbb{Z}_r^z\sigma_{j+r}^y-\sigma_j^y\mathbb{Z}_r^z\sigma_{j+r}^x)\Bigg ]-\frac{h'}{2}\sum_{j=1}^{N}\sigma_j^z,
\end{eqnarray}
where $\mathbb{Z}_r^z = \prod_{l=j+1}^{j+r-1}\sigma_l^z$, with $\mathbb{Z}_1^z=\mathbb{I}$
and \(\sigma^k\)(\(k=1,2,3\)) is the Pauli matrix, \(J_r'=\frac{J}{r^\alpha}\) with \(\alpha\) being the strength of power-law decay of the model, $\gamma$ is the anisotropy parameter, $D$ is the strength of the DM interaction, \(\mathcal{N}=\sum_{r=1}^{N/2}\frac{1}{r^\alpha}\) is called the Kac-scaling factor \cite{Kac_jmp_1963} which ensures extensivity of the energy in the case of finite-size systems \(\alpha<1\) in the finite-size limit and \(h'\) is the strength of the external magnetic field. To make the analysis dimensionless, we redefine the magnetic field as \(h=h'/J\) and DM interaction strength as  \(D=D'/J\). We consider the periodic boundary condition (PBC), i.e., \(\sigma_{N+1}\equiv \sigma_1\).  This model can be solved analytically by mapping spins into free fermions under the following Jordan-Wigner transformation \cite{barouch_pra_1970_1, barouch_pra_1970_2, lieb1961, glen2020}:
\begin{align}
    \sigma^x_n &=  \left( c_n + c_n^\dagger \right)
 \prod_{m<n}(1-2 c^\dagger_m c_m) \nonumber \\
    \sigma^y_n &=i\left( c_n - c_n^\dagger \right)
 \prod_{m<n}(1-2 c^\dagger_m c_m) \nonumber \\\text{and}\quad
\sigma^z_n&=1-2 c^\dagger_n  c_n,
 \label{eq:Jordan_wigner}
\end{align}
where \(c_m^\dag\)(\(c_m\)) is creation (annihilation) operator of spinless fermions and they follow fermionic commutator algebra. The corresponding free fermionic version of the Hamiltonian takes the form 
\begin{eqnarray}
   \nonumber H&=&\sum_{n}\sum_{r}\frac{J_r}{2}\left ((1+iD)c_n^\dagger c_{n+r}+\gamma c_n^\dagger c_{n+r}^\dagger + \text{h.c}\right )\\&&+h(c_n^\dagger c_n-1/2),
   \label{eq:JW_hamil}
\end{eqnarray}
where \(J_r=\frac{1}{r^\alpha}\). One can observe that the Hamiltonian which is quadratic in fermionic operators have long-range interaction both in hopping and superconducting terms. In addition, an extra hoping term emerges due to the DM interaction which also induces a complex phase. It is worthwhile to mention here that in the absence of DM interaction for \(\gamma=0\), the model is \(\mathbb{U}(1)\) symmetric while it is \(\mathbb{Z}_2\) symmetric with \(\gamma\ne 0\), \(D=0\). The introduction of  DM interaction leads to a competition between these two symmetries.   In the short-range (SR) models, it is known that the conformal sector changes due to the shift of the symmetry sector \cite{Zhong_2013}, responsible for a non-trivial phenomenon.

Due to the presence of translational invariance in the system, momentum is a good quantum number. Hence, we perform a Fourier transform of the form, \(c_n=\frac{1}{\sqrt{N}}\sum_k e^{-i\phi_kn}c_k\) where \(\phi_k=\frac{\pi (2k-1)}{N}\) \(\forall k\in[-N/2,N/2]\) and the corresponding Hamiltonian in the momentum basis can be written as \(H=\underset{k}{\oplus} H_k\) with 
\begin{eqnarray}
    H_k &=& \sum_{k>0}-\Big [\sum_{r}(\text{Re}(J_k^\alpha)+D\text{Im}(J_k^\alpha))c_k^\dagger c_k\nonumber\\&& +(\text{Re}(J_k^\alpha)-D\text{Im}(J_k^\alpha))c_{-k}^\dagger c_{-k}\nonumber\\&&-i\gamma\text{Im}(J_k^\alpha) (c_k^\dagger c_{-k}^\dagger- c_{-k} c_{k})\Big]\nonumber\\&& +h(c_k^\dagger c_k+c_{-k}^\dagger c_{-k})-h, 
    \label{eq:monentum_H}
\end{eqnarray}
and $J_k^\alpha = \sum_{r = 1}^\frac{N}{2}J_r e^{i\phi_kr}$. This Hamiltonian can be diagonalized by applying the Bogoliubov transformation,
\begin{equation}
    \begin{pmatrix}
        c_k\\c_{-k}^\dagger
    \end{pmatrix}=\begin{pmatrix}
        u_k & -v_k^{*}\\
        v_k & u_k^{*}
    \end{pmatrix}\begin{pmatrix}
        \tau_k\\ \tau_{-k}^\dag
    \end{pmatrix},
\end{equation}
where \(u_k=\cos\theta_k\), \(v_k=i\sin\theta_k\) and \(\theta_k\) is called Bogoliubov angle, given by \(\cos\theta_k=\frac{-h+\text{Re}(J_k^\alpha)-\lambda}{\sqrt{2(\lambda^2-(\text{Re}(J_k^\alpha)-h)^2)}}\) and \(\lambda=\sqrt{(h-\text{Re}(J_k^\alpha))^2+(\gamma~\text{Im}(J_k^\alpha))^2}\). The eigenvalue corresponding to each momentum, \(k\), is given as
\begin{equation}
    \epsilon_{\pm k}=-D~\text{Im}(J_k^\alpha)  \pm \lambda.
    \label{eq:dispersion}
\end{equation}
Note first that for \(D\ne 0\), although the Bogoliubov basis, \((u_k, v_k)^T\), are independent of DM interactions,  \(\epsilon_k\ne \epsilon_{-k}\),  making the analysis non-trivial.
%Hence, in order to make easy our calculation, in some of our calculation we follow the method provided in Barouch {\it et.al.} \cite{barouch_pra_1970_1, barouch_pra_1970_2}.
Further, when \(D=0\), this model shows criticality at \(h=1\) and at \(h=-1+2^{1-\alpha}\) for \(\alpha>1\)  corresponding to the both ends of the Brillouin zone, i.e., \(\phi_k=0\) and \(\phi_k=\pi\) respectively.  In the case of \(\alpha<1\), the energy of the system diverges in the thermodynamic limit, hence, we consider large but finite system size \cite{solfanelli_jhep_2023}. With the introduction of DM, new gapless region emerges, which changes with \(\alpha\) for \(D>\gamma\). Such a region is known for the  nearest-neighbour $XY$ model when \(D>\gamma\) \cite{roy_prb_2019}  and  gapless to gapped transition occurs at \(D= \gamma\). However, we now establish that even when \(D>\gamma\), long-range interactions can induce a gapped phase instead of a gapless one, which implies the modification of the critical lines in presence of LR interactions.

%{\color{blue}\subsubsection*{Possible experimental implementation}}

\textcolor{black}{{\it Possible experimental implementations of LR model with DM interaction.} Long-range interacting systems occur naturally in trapped-ion experimental platforms, providing a highly regulated set-up for mimicking complicated quantum many-body phenomena \cite{trapped_ion_1,trapped_ion_2,trapped_ion_3}.  Engineering spin-dependent optical dipole forces allows an exact tuning of the interaction range, as defined by the parameter $\alpha$.  This tunability enables the investigation of many interaction domains, ranging from practically infinite-range interactions to regimes more akin to short-range behavior.  Furthermore, the strength of an effective transverse magnetic field, represented by $h$, can be independently regulated by adjusting the trapping potential in Penning trap configurations.  This control is necessary for simulating a variety of quantum magnetic models, including those with quantum phase transitions and critical behaviors \cite{quantum_magnetism_ion_trap}.}

\textcolor{black}{Recent studies have focused on the realization of the Dzyaloshinskii-Moriya interaction, which is quantified by the parameter $D$. According to a recent suggestion \cite{athreya_bilayer_iontrapp_prx_2024}, a bilayer Penning trap structures may support complex eigenvectors and build effective spin-orbit couplings, making them a feasible way to mimic DM interactions. These studies suggest that a ion-trap experiment in a bilayer Penning trap can effectively simulate DM interaction with the ability to modulate all the parameters such as $\alpha$, $h$ and $D$. Furthermore, there are several trapped ion experiments where long-range Ising interactions have been realized, providing a platform to realize  long-range DM interactions \cite{Jurcevic2014Jul}. }

\textcolor{black}{Further,  it has been suggested that Kitaev chains with long-range hopping and pairing can serve as models for helical Shiba chains, consisting of magnetic impurities in an s-wave superconductor. The introduction of an extended DM interaction results in local phases for this interaction, adding hopping terms that generally involve complex phase factors, which may be simulable in such chains 
%Therefore, it may be possible to simulate such interactions in this system 
\cite{pientka_prb_2013}.  The experimental signature of the DM interaction is the presence of skyrmions and the chiral nature of magnon quasiparticles. Since the DM interaction arises due to the broken inversion symmetry and strong spin-orbit coupling, it could potentially be realized in spherical FeGe crystals \cite{sinaga_prb_2024}.}

\begin{figure}
    \centering
    \includegraphics[width=\linewidth]{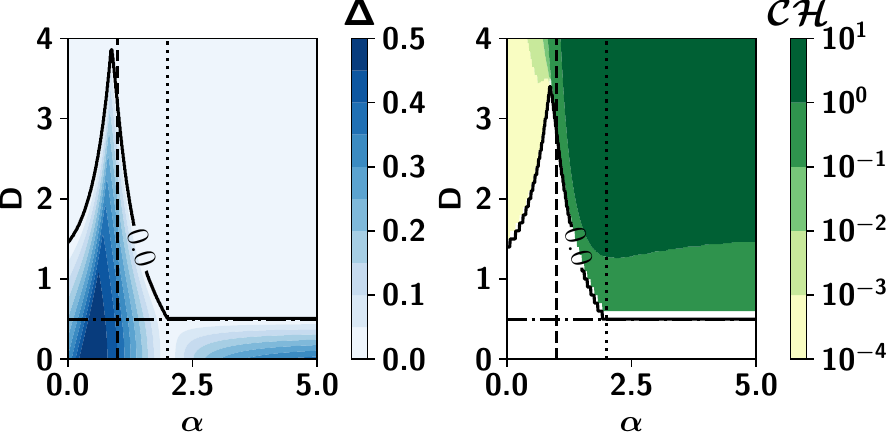}
    \caption{\textbf{Contour plot of \(\Delta\) and \(\mathcal{CH}\) in the  \((\alpha,D)\)-plane.}  \(D\) is the interaction strength of the  Dzyaloshinskii-Moriya interaction while \(\alpha\) indicates the power-law decay in the range of interactions.  Here the anisotropy parameter, \(\gamma =0.5\) and the strength of the magnetic field, \(h=-0.5\). For the NN \(XY\) model, it is known that the system is gapless when \(D> \gamma =0.5\). However, we find that this is not the case when \(1< \alpha <2\) and \(\alpha<1\). Moreover, the right-hand sinde figure shows that whenever the system is gapless, it is chiral, i.e., \(\mathcal{CH} \neq 0\). The system size is chosen to be \(N=512\). All axes are dimensionless.}
    \label{fig:chiral_phase_alpha}
\end{figure}
%\section{Phase diagram and change in decay of correlations via LR asymmetric interactions}
\label{sec:phase_classification}

\begin{figure}
    \centering
    \includegraphics[width=\linewidth]{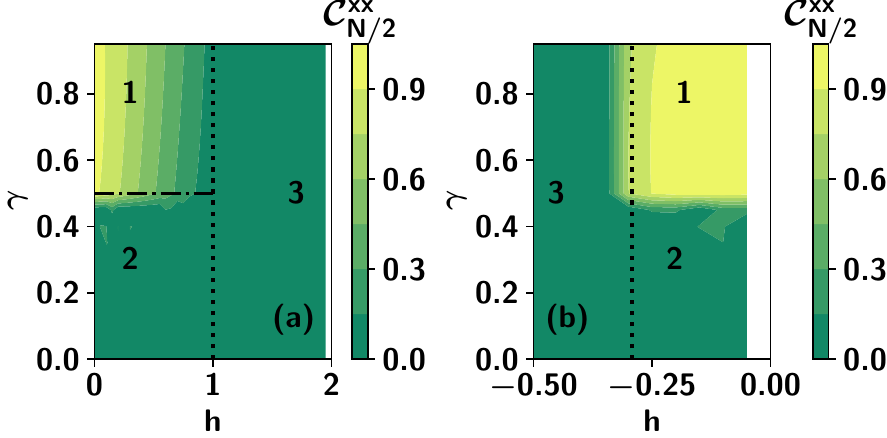}
    \caption{\textbf{Phase diagram of the long-range extended \(XY\) model with nonvanishing DM interactions.}  Nonvanishing \(\mathcal{C}_{N/2}^{xx}\)  indicates the \(FM_x\) phase for varying external magnetic field, \(h\) and the anisotropy parameter, \(\gamma\) in the \(xy\)-plane. For \(h>0\), the critical point is independent of \(\alpha\), \(FM_x\) extends up to \(h=1\) and bounded by \(D=\gamma\) -line in the \(y\)-direction (see the shaded region Fig (a)). On the other hand, for \(h<0\), critical points depends upon \(\alpha\) and \(FM_x\) phase is bounded upto \(h=-1+2^{1-\alpha}\) in the \(x\)-direction. But the bound in the \(y\)-direction is not easy to derive which depends non-linearly on \(\alpha\) (see the shaded region Fig (b)).  Other parameters of the systems are \(D=0.5\), \(\alpha=1.5\) and \(N=512\). All the axis are dimensionless.}
    \label{D=0.5}
\end{figure}

% \begin{figure}
%     \centering
%     \includegraphics[width=\linewidth]{energy_gap.pdf}
%     \caption{ \textbf{Contour plot of \(\Delta\) with \(\alpha\) and \(D\).} We plot the gap detector for long-range DM model for a constant value of external field \(h=0.5\) in (a) and \(h=-0.5\) in (b). Gapped region get extended when \(h<0\) which showcase the dependency of criticality on \(\alpha\) (see. Fig (b)). In the gapped region, extended gapped region is highly suppressed, only high \(\alpha\) value can produce extended gapped region.  Values of other parameters are \(\gamma=1\) and \(N=1024\). All axes are dimensionless.}
%     \label{fig:energy_gap}
% \end{figure}

\subsection{The gapless chiral phase depending on LR interactions}
\label{sec:gapless_chiral}
Apart from quantum phase transitions at zero temperature in which the energy gap vanishes, the LR Ising models possess another three distinctive regimes with respect to power-law fall-off rate \(\alpha\), namely, non-local $(0 < \alpha < 1)$, quasi-local $(1 < \alpha < 2)$ and local $(\alpha > 2)$ according to different scaling law of correlation length \cite{Vodola_prl_2014,Vodola_njp_2015}. It is important to note here that in the absence of DM interaction, gap-closing never occurs by tuning the parameter $\alpha$ for \(h>0\) except \(h_c=1\). Further, we know that \(\alpha\gg 2\), the system with NN interactions becomes gapless for \(D>\gamma\) depending on \(h\) and the chiral phase emerges along with paramagnetic and ferromagnetic phases. Due to the long-range DM interactions, a competition between the range of interactions, \(\alpha\) and the strength of DM interactions surfaces which becomes responsible for the change in condition for the gapless  to gapped transition. More specifically, for moderate values of \(\alpha(<2)\), \(D>\gamma\) does not guarantee chiral gapless phase. We find the following in case of LR model. 

%Let us illustrate here some interestning transition which appears due to the competetion between LR interactions, \(\alpha\) and strong DM interactions, \(D>\gamma\). Specifically, we obtain the following:

\begin{proposition}
    For a given value of the magnetic field $h$ (both in \(h\ge 0\) and \(h<0\)), a gapped to gapless transition occurs when  $D>\gamma$ depending upon the value of \(\alpha\) instead of \(D=\gamma\), known for the NN model.
\end{proposition}

\begin{proof}
    The proof is done by analyzing the dispersion relation in Eq. (\ref{eq:dispersion}). A system is said to be
% Looking at the dispersion relation in Eq. (\ref{eq:dispersion}), one finds that the spectrum can be gapless through out the range of parameters present in Hamiltonian. More specifically, due to the fact that \(|\epsilon_k|\ne |\epsilon_{-k}|\) in presence of the DM interaction,  an extended gapless phase is induced. The 
gapped when the expression, \(\Delta\) \cite{soltani_magn_2019,luo_prb_2022}, defined as
\begin{equation}
    \Delta=\max\{\min\{\epsilon_k\},0\}>0;\quad \forall \phi_k\in [-\pi, \pi], 
\end{equation} 
otherwise, the spectrum is gapless (with $\Delta = 0$). The exact point, where the gap-closing takes place, depends on the parameters of the Hamiltonian, and the corresponding momentum, \(k_F\), called the Fermi point, can be obtained as a solution of \(\pdv{\epsilon_k}{k}=0\). However, due to the presence of long-range order, it is cumbersome to find the Fermi point analytically. By differentiating numerically, we observe that the gapless phase never occurs when $D<\gamma$ while for $D>\gamma$, the gapless phase occurs although it is not ubiquitous, and changes with the variation of \(h,D,\gamma\) and \(\alpha\). In particular, when \(h<0\), with the decrease of \(\alpha\), the system is gapped even when \(D>\gamma\) and two distinctive gapped regions emerge when \(1\le \alpha<2\) and \(0<\alpha<1\) which depend on the strength of \(D\), thereby inducing (destroying) gapped (gapless) regions. However, for \(\alpha>2\), $D > \gamma$ guarantees gapless phases and the gapless to gapped transition occurs at \(D=\gamma\). On the other hand, we find, that when \(1\le \alpha<2\), strong DM interactions are required to obtain the gapless phase and the desired asymmetric interaction strength for obtaining the gapless phase {\it increases} with the decrease of \(\alpha\) (see Fig. \ref{fig:chiral_phase_alpha}). Note, further, that the maximum \(D\) necessary for the gapless phase does not happen exactly at \(\alpha=1\), but at a nearby point which can be due to finite size analysis.  

\end{proof}

It is now evident that the interplay of \(\alpha\) and \(D\) along with \(\gamma\) and \(h\) can change the critical lines of the phase diagram. Therefore, we now examine the phases in the \((h,\gamma)\)- and \((\alpha, D)\)-planes.

{\it $\alpha$-dependent chiral phase.} In the nearest-neighbor  $XY$ model with DM interactions, the gapless phase gives rise to a chiral order that measures local current flow between the nearest neighbor sites. It can be established, by computing \cite{jafri_prb_2008} the order parameter,
\begin{equation}
    \mathcal{CH}= \expval{\frac{1}{4N}\sum_{i=1}^N \sigma^x_i\sigma^y_{i+1}-\sigma^y_i\sigma^x_{i+1}}
\end{equation}
between any two neighboring sites, \(i\) and \(i+1\) in the ground state which describes the helical alignment of spins or chiral order in the \(z\)-direction. 
% Owing to the periodic boundary condition, it is enough to find the chirality of one of the nearest neighbors. \textcolor{red}{We further study the effect of chirality between distance neighbourhood as there there may exist current flow between distance sites due to long-range interaction.} 
% We observe that for every value of \(\alpha\) there is a particular DM interaction strength that induces a gapless phase (see Fig. \ref{fig:energy_gap}(a)). Also in presence of strong long-range interaction strength, there is a trade off between \(\Delta_k\) and the strength of DM interaction \(D\), such a scenario is non-existent in the case of nearest-neighbour model XY model with DM interaction. 
In the nearest-neighbor case, i.e., for the NN $XY$ model with $D > \gamma$ and suitable \(h\) values, it was shown that $\mathcal{CH} > 0$ when $\Delta = 0$. 

In the LR model, we find that for a fixed $h$, $\alpha <2$ and $D > \gamma$, whenever $\Delta = 0$, i.e., the system is gapless, it possesses chiral order with order parameter $\mathcal{CH} > 0$ (as shown in Fig. \ref{fig:chiral_phase_alpha}). Interestingly, however, there exists region in which for \(D> \gamma\), \(\Delta>0\) (as shown in the above proposition), and the corresponding chiral order parameter vanishes, i.e., \(\mathcal{CH} =0\).

% when there is a gapless phase in the presence of DM interaction, the chiral phase is present, i.e., $\mathcal{CH}>0$ although the converse is not true, i.e., wherever the chiral phase is present, it is not necessary that the underlying spectrum is gapless \cite{jafri_prb_2008}. We observe the former effect in case of long-range interaction, while the latter is absent due to the Ferromagnetic nature of the interaction. Specifically, for a fixed \(h\) and \(\alpha\) the model with \(D<\gamma\) is gapped and $\mathcal{CH}$ calculated in the ground state vanishes in this region while we observe that for a fixed \(h\) and \(\alpha>1\), when \(D>\gamma\) and the system is gapless, the value of chiral order parameter, $\mathcal{CH}$.
%We consider the XY model with non-unity anisotropy parameter (\(\gamma=1\)) and find the gapless nature of the spectrum and find that the system is gapped when $D<\gamma$ as expected but it is still gapped for several values of $D>\gamma$ due to the presence of strong power law interaction (see Fig. \ref{fig:mag_figure}(a)). We have observed the similar phenomenon in the case of the chiral order parameter. We find that when the system is gapless and $\alpha > 1$, the non-zero chirality order parameter emerges within the same sector with significant value of chirality, on the other hand, when $\alpha < 1$, the chirality value is very minimal (see Fig. \ref{fig:mag_figure}(b)). 

{\it Ferromagnetic and paramagnetic phases.} We now calculate the long-range magnetic order parameter, $\mathcal{C}^{xx}_{N/2} =  \expval{\sigma^x_1\sigma^x_{N/2}}$. When $\mathcal{C}^{xx}_{N/2} \ne 0$, the system is said to be in the ferromagnetic-\(x\) phase $(FM_x)$ while paramagnetic when it vanishes. On one hand, we show that the $FM_x$ phase occurs when the system is gapped and $0<h<h_c^1=1$, $1<\alpha<2$, and $D < \gamma$ (see Fig. \ref{D=0.5}). On the other hand, when $h$ is negative, the system becomes ferromagnetic, i.e., $\mathcal{C}^{xx}_{N/2} > 0$  when $h > h_c^2=-1+2^{1-\alpha}$ with \(1<\alpha<2\). When the parameters  are neither chiral nor ferromagnetic, the system is in a paramagnetic phase. 

{\color{black} Tuning the strength of the parameters, we can qualitatively describe  the phases of this model which changes due to the introduction of long-range interactions. The strength of long-range interactions, \(\alpha\), perturbs the system, and  can delay the transition from a gapped to a gapless region, although no analytical phase boundary has been established so far. We provide the phase boundaries based on numerical simulations, which depend on (\(\alpha\)), \(\gamma\), and \(h\). More precisely, we present the qualitative phase boundaries for this extended long-range DM model.
\begin{enumerate}
    \item  \textbf{ \(D > \gamma\) and \(\alpha > 2\)}: The system is in the gapless chiral phase which also depends upon the strength of the magnetic field \(h\), although, the exact boundary is not known.

\item \(\forall h\), \textbf{ \(D \leq \gamma\) and \(\alpha > 2\)}: The system is in the gapped region.
    
    \item \textbf{\(h > 0\) \text{or} \(h<0\), \(D > \gamma\) and \(0 < \alpha < 2\)}: The system may belong to the gapless chiral region, but depending on the strength of \(\alpha\),  there exists a gapped region for a fixed \(h\), which is an artifact of the long-range interaction strength.

\end{enumerate}
}

\subsection{Constant decaying exponent of correlations with DM interactions}
\label{sec:static_corre_decay}

From the above analysis, it is evident that the transition from the chiral gapless phase to the other magnetic gapped phase heavily depends on the fall-off rate $\alpha$. Let us now investigate how the scaling of classical correlations and quantum mutual information between two spins (a measure of total correlations containing both quantum and classical correlations components) \cite{Winter05} can recognize the $\alpha$-dependent phases discussed above. 
% As mentioned earlier, the long-range interacting model without the presence of DM interaction contains distinctive transition points characterised by the decay rate of both classical and quantum correlations \cite{}. Motivated from such a study, we perform the study on how classical correlations decay when the distant between spins increases. In addition, we also check how mutual information which contains the information about both classical and quantum correlation present in the system decays with the distance. 
The classical correlation (CC) between two sites, separated by a distance $R$,  can be represented as $\mathcal{C}_R \equiv \mathcal{C}^{xx}_{i,i+R} = \langle \sigma_i^x \sigma_{i+R}^x \rangle$ while the quantum mutual information is defined as 
$\mathcal{I}_R \equiv \mathcal{I}_{i:i+R}=S(\rho_i)+S(\rho_{i+R})-S(\rho_{i, i+R}),$ where $S(\sigma) = -\tr(\sigma\log_2\sigma)$ indicates the von-Neumann entropy of the state \(\sigma\), $\rho_i$ and $\rho_{i+R}$ are the reduced density matrices of the joint state $\rho_{i,i+R}$. As we impose periodic boundary condition, the site index $i$ can be ignored and both the classical and total correlation can be studied as $\mathcal{C}^{xx}_R$ and $\mathcal{I}_R$, where $R = \{1,2,\ldots \frac{N}{2}\}$. In typical one-dimensional (1D) quantum spin models, the CC decays exponentially when it is away from criticality, whereas, at the criticality, the polynomial decay with $R$ is observed. However, the extended Ising model deviates from this norm by showing different scaling laws near the critical points and away from it \cite{sadhukhan_arxiv_2021}. Moreover, in the long-range Kitaev chain (see Eq. (\ref{eq:JW_hamil}) with $D = 0$), quantum mutual information (QMI) persists between two distant segments of the chain, provided the system is in the non-local regime, i.e., $\alpha < 1$, thereby  predicting the correlation between two distant regions of the system  \cite{francica_prb_2022}.

In our study, we expect to have a notable scaling law due to the presence of long-range and additional DM interactions which is indeed the case. Firstly, we observe that the trends of $\mathcal{C}_R^{xx}$ and $\mathcal{I}_R$ with $R$ is qualitatively different when the system is in the chiral phase and when it is not. Secondly, a counter-intuitive observation is that if one tunes the parameters 
\begin{figure}
    \centering
    \includegraphics[width=\linewidth]{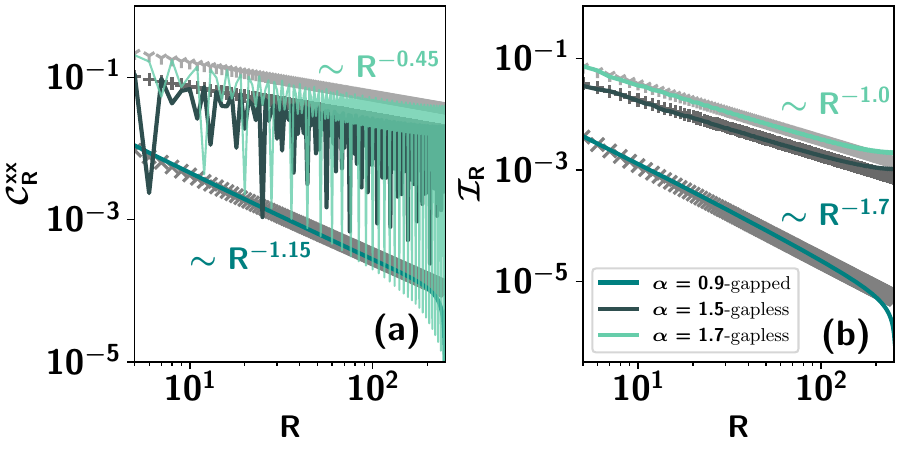}
    \caption{Classical correlation $(\mathcal{C}^{xx}_R)$ and mutual information $(\mathcal{I}_R)$ (ordinate) of the bipartite state $\rho_R$ obtained from the ground state vs the distance $R$ (abscissa) between two arbitrary sites, \(i\) and \(i+R\). The corresponding fit is drawn with a dashed line in the same color. The parameter set of $h=-0.5$ and different values of $\alpha$ indicate whether the system is in a gapless or gapped phase. All other specifications are same as in Fig. \ref{fig:chiral_phase_alpha}. All the axes are dimensionless. }
    \label{fig:mu_decay_chiral}
\end{figure}
 \(D\), \(\gamma\), \(h\) and \(\alpha\) in such a way that \(\mathcal{CH}\ne 0\), the decay exponents of mutual information and classical correlation become constant as depicted in Fig. \ref{fig:mu_decay_chiral}. Specifically, we observe that irrespective of the details of the power-law exponent and the strength of the magnetic field, the decay exponent of CC and QMI falls off as 
\begin{align}
    \mathcal{I}_R & \propto R^{-1} \quad \text{and} \\
    \mathcal{C}^{xx}_R &\propto R^{-0.45},
\end{align}
which are independent of \(\alpha\) provided the ground state belongs to the gapless chiral region (see Fig. \ref{fig:mu_decay_chiral} and Table \ref{tab:xyz_training_study_results}). This behavior of \(\mathcal{I}_R\) and \(\mathcal{C}^{xx}_R\) highlights a more complex and highly correlated patterns between \(\alpha\) and \(D\) within the gapless region. It indicates that although the range of interactions can control the critical lines in the $(D,\gamma)$-plane, differentiating gapless and gapped phases, the decay rate of CC and QMI in the gapless phase remains unaltered  and the chiral phase favors the sharing of both quantum and classical correlations between distant sites. More importantly, such $\alpha$-dependent scaling behavior holds both for quantum and classical correlations (see Fig. \ref{fig:mu_decay_chiral}).  

On the other hand,  when the system is in the non-chiral (which we refer to as achiral) phase, the ground state displays the decay exponent of both classical and quantum correlations which vary with $\alpha$, i.e., depending on the internal descriptions of the model (see Fig. \ref{fig:mu_decay_chiral} and Table \ref{tab:xyz_training_study_results}). 

Through these observations, we elucidate the nuanced role of long-range DM interactions in modulating correlation properties across both chiral and achiral regions leading to the following proposition.

\begin{proposition}
    In the chiral region of the LR $XY$ model with DM interactions, the scaling exponents of CC and QMI are constant (independent of the fall-off rate, $\alpha$) while their exponents depend on $\alpha$ in the achiral region.
\end{proposition}

\begin{figure}
    \centering
    \includegraphics[width=\linewidth]{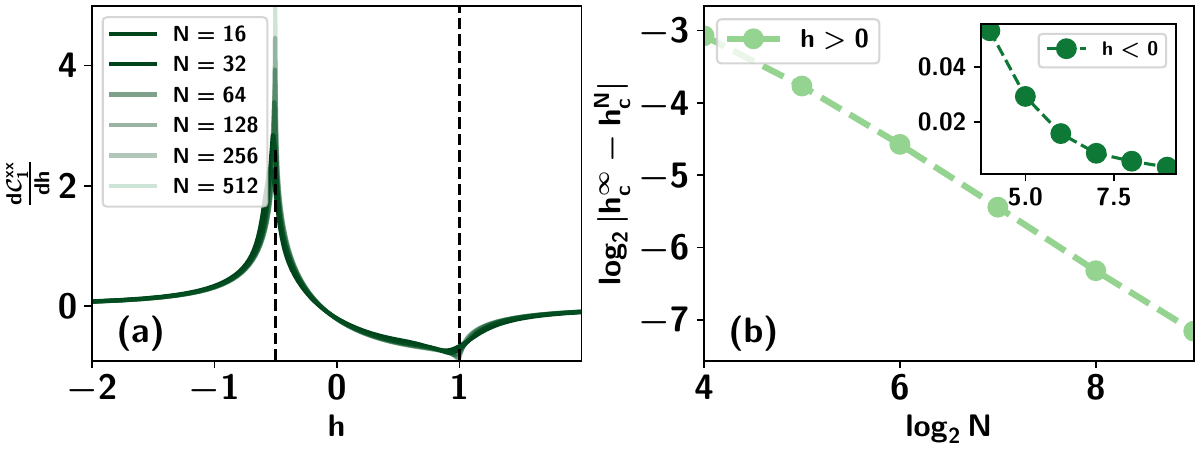}
    \caption{\textcolor{black}{\textbf{Finite size scaling, using nearest-neighbor classical correlation, \(\mathcal{C}_1^{xx}\)}. (a) A kink in the derivative of \(\mathcal{C}_1^{xx}\) indicates the critical points for both \(h > 0\) and \(h < 0\). For \(h > 0\), the critical point is independent of \(\alpha\), and the first derivative exhibits a kink at \(h_c = 1\), which becomes sharper as the system size \(N\) increases. (b) To perform the scaling analysis, we plot \(\log_2|h_c^{\infty} -h_c^N|\) against \(\log_2 N\) where \(h_c^{\infty} =1\) and 
    \(h_c^N\) is the value at which \(\mathcal{C}^{xx}_1\) shows the kink for a given system size \(N\). Its behavior indicates that \(N = 512\) provides a good approximation to the thermodynamic limit. On the other hand, for \(h < 0\), the critical point depends on \(\alpha\), and is given by \(h = -1 + 2^{1 - \alpha}\). This critical point exhibits the same trend as the one for \(h > 0\), becoming sharper with increasing system size (see the inset of (b)). The system parameters used are \(D = 0\) and \(\alpha = 2\). All axes are dimensionless.}
}
    \label{fig:scaling_N}
\end{figure}

\textcolor{black}{{\it Scaling analysis.} To provide insights into the thermodynamic limit, we perform a scaling analysis to understand how the system approaches criticality as the system size increases. Specifically, we analyze \(\log_2 |h_c^\infty - h_c^N|\), obtained by computing nearest-neighbor classical correlation \(d\mathcal{C}_1^{xx}/dh\), as a function of \(\log_2 N\), where \(h_c^\infty\) and \(h_c^N\) are the critical points in the thermodynamic limit and when the system size is \(N\) respectively. The decay of this quantity with \(\log_2 N\) indicates that the quantity reaches the thermodynamic limit. Furthermore, the model exhibits critical points at \(h_c = 1\) and \(h_c = -1 + 2^{1 - \alpha}\) for \(\alpha > 1\). We demonstrate that the classical correlations in the \(x\)-direction capture these critical points, and as the system size increases, the identification of these critical points becomes more prominent, as shown in Fig.~\ref{fig:scaling_N}.  These results indicate that the classical correlations with system size \(N=512\) indeed mimics the properties with \(N \rightarrow \infty\) which we again confirm by computing \(\mathcal{C}^{xx}_1\) with system size upto \(N=1024\). On the other hand, when calculating the dynamical correlations, we can extend the system size up to \(N = 30K\)  (see Sec. \ref{sec:dynamics}).  Note, however,
it is difficult to provide closed form expressions for the correlators owing to the presence of complicated integrals which is true even for short-range models. 
}

\begin{table}[]
    \centering
    \begin{tabular}{|c| c | c | c | c | c |}
    \Xhline{2\arrayrulewidth}
    \Xhline{2\arrayrulewidth}
    D & h & \(\alpha\) & $\mathcal{I}_R$ & \(\mathcal{C}^{xx}_R\) & Phase\\
    \Xhline{2\arrayrulewidth}
    \Xhline{2\arrayrulewidth}
     $1.5$ & $0.5$ & $ 0.5 $ & \(R^{-0.25}\) & \( R^{-0.3}\) & \text{ gapped} \\
   % \Xhline{2\arrayrulewidth}
    $1.5$ & $0.5$ & $0.8$ & \(R^{-1}\) & \(R^{-0.45}\) & \text{ gapless} \\
    $1.5$ & $0.5$ & $1.3$ & \(R^{-1}\) & \( R^{-0.46}\) & \text{ gapless} \\
\Xhline{2\arrayrulewidth}
    \
     $2.5$ & $-0.5$ & $0.5$ & \( R^{-0.4}\) & \( R^{-0.5}\) & \text{ gapped} \\
   % \Xhline{2\arrayrulewidth}
    $2.5$ & $-0.5$ & $0.8$ & \(R^{-1.5}\) & \(R^{-1.2}\) & \text{ gapped} \\
    $2.5$ & $-0.5$ & $1.3$ & \( R^{-1}\) & \(R^{-0.45}\) & \text{ gapless} \\
    \Xhline{2\arrayrulewidth}
    \end{tabular}
    \caption{The scaling laws of mutual information \(\mathcal{I}_R\) and classical correlation \(\mathcal{C}^{xx}_R\) with \(R\) for various parameter values of DM interaction strength ($D$), magnetic field ($h$), and power-law exponent ($\alpha$). When the system is gapless, universal scaling exponents for both the quantities emerge which is not the case for the gapped regimes.  }
    \label{tab:xyz_training_study_results}
\end{table}

\subsection{Amendment in scaling of entanglement entropy and central charge owing to long-range DM interactions}
\label{sec:entropy_charge}

In a many-body system, an insight of universality classes can be obtained from the generic features of the ground state. A prominent physical quantity that serves the purpose is the block entanglement entropy  of block-size $l$ \cite{fazio_rev, Calabrese2004Jun}. In particular, after partitioning the $N$-party ground state into two blocks, containing $l$ and $N-l$ sites (with $l \ll N$), we compute the block entanglement of the ground state as $S_l = S(\rho_l)$ with $\rho_l$ being the reduced density matrix of the \(N\)-party ground state.  The $l$-block entropy can be calculated from the two-point correlation functions between modes, \(m\) and \(n\), given as
\begin{equation}
    C_{mn}=\langle \Psi|c_m^\dagger c_n | \Psi \rangle, \quad F_{mn}=\langle \Psi|c_m^\dagger c_n^\dagger |\Psi \rangle,
\end{equation}
where \(\ket{\Psi}\) is the ground state and the corresponding correlation matrix, consisting of correlation functions, can be written as
\begin{equation}
    \mathbb{C}_l=\begin{pmatrix}
        \mathbb{I}-C & F\\
        F^\dagger & C
    \end{pmatrix}.
\end{equation}
If \(\lambda_p\)'s are the eigenvalues of \(\mathbb{C}_l\), the von-Neumann entropy of \(\ket{\Psi}\) in the bipartition \(l:N-l\) is given as \(S_l=-\sum_{p=1}^{2l}\lambda_p\log_2\lambda_p\).
% One of such feature is von-Neumann entropy which is defined as \(S_l=- which measure the entanglement present between different blocks of a pure state.
If a system is gapped which typically happens away from the criticality and when the range of interactions is relatively local, i.e., the norm of the Hamiltonian does not increase as the system size increases, $S_{l}$ follows the area law, which implies $S_{l}^{NC} \propto l^{d-1}$, where $d$ is the spatial dimension of the system and $NC$ stands for ``not at criticality" while it deviates from the area law at criticality \cite{Calabrese2004Jun}. Hence from the scaling behavior of \(S_l\), one can detect the transition from a gapped to a gapless phase which is interesting since a simple scalar quantity can describe the essential properties of the Hamiltonian instead of a complete microscopic description. 

On the other hand, when the system becomes gapless, representing the critical point, e.g., in the transverse NN $XY$ model,  $S_l$ scales logarithmically, i.e., for a translationally invariant spin chain, at critical points, \(S_{l}^C= \frac{c}{3}\log_2\Big[\frac{N}{\pi}\sin\frac{\pi l}{N}\Big]+a\) \cite{holzhey_npb_1994, Calabrese2004Jun}, where \(c\) is the {\it conformal or effective central charge} in the conformal field theory and \(a\) is a non-universal constant. Further, it was shown that the central charge carries the signatures of an underlying symmetry. For example, \(c\) takes value \(1/2\) for the transverse-field Ising-like Hamiltonian at the critical point for which the Hamiltonian adheres to the \(\mathbb{Z}_2\)-symmetry, while with the addition of asymmetric nearest-neighbor DM interaction, the Hamiltonian becomes  gapless region having $c=1$ and gapped with $c=1/2$. Hence EE turns out to be a powerful tool to separate a gapped phase from a gapless one.
%even though it is a simple scalar quantity, carrying essential properties of the Hamiltonian. 

It was also recently shown that if one considers long-range interactions instead of short-range ones \cite{yang_arxiv_2024,chakraborty_arxiv_2024}, where magnetic criticality depends on the value of $\alpha$, EE depends upon \(\alpha\) and the free-fermionic version of long-range Ising spin shows fractal entanglement apart from the volume and area law \cite{solfanelli_jhep_2023} and the scaling of entropy at critical points may not follow \(S_{l}^C\). In this case, we can define an effective conformal charge, \(c_{eff}\), to describe the universal properties of the system.

We are interested in finding how the scaling of $S_l$ in the LR Ising Hamiltonian gets altered in the presence of DM interactions. Specifically, we will show the revision of a central charge $c_{eff}$, which occurs in $S_l^C$. Before presenting the effects of DM interaction on the central charge, we observe that at criticality, i.e., $h_c^1 = 1$ of the extended Ising model, $c_{eff}$ monotonically increases when $1 < \alpha < 2$ and saturates to $1/2$ at $\alpha \ge 2$, (which is the scaling known for the NN model) as shown in Fig. \ref{fig:ceff}(a). When $\alpha < 2$, $c_{eff}$ varies non-linearly with $\alpha$ which points out that due to the presence of long-range interaction, conformal symmetry breaks down. 

On the other hand, the introduction of long-range  DM interaction makes more modifications in $c_{eff}$ -- \(\,\) (i) the saturation value changes to $1$ when $\alpha \rightarrow \alpha_c = 2$, with $\alpha_c$ being the point which mimics the NN transverse Ising model with NN DM interactions; (ii) for \(\alpha<\alpha_c\), the conformal symmetry breaks down whereas when $\alpha > \alpha_c$,  $c_{eff}\rightarrow 1$ (see Fig. \ref{fig:ceff}(b)). 
Our study reveals that even with the introduction of long-range asymmetric interaction, the conformal symmetry is not restored for all the values of $\alpha$ although $c_{eff}$ indicating the conformal invariance begins for smaller values of $\alpha_c<2$ which is found in the absence of DM interactions. This behavior can again be attributed to the fact that when the DM interaction is present, the system becomes gapless upon varying the power-law decaying factor $\alpha$.

% Now, with the introduction of the asymmetric interaction, this system gets even more complex and as discussed above the gapped to gapless transition depends on $\alpha$ when $h$ is negative. 
\begin{figure}
    \centering
    \includegraphics[width=\linewidth]{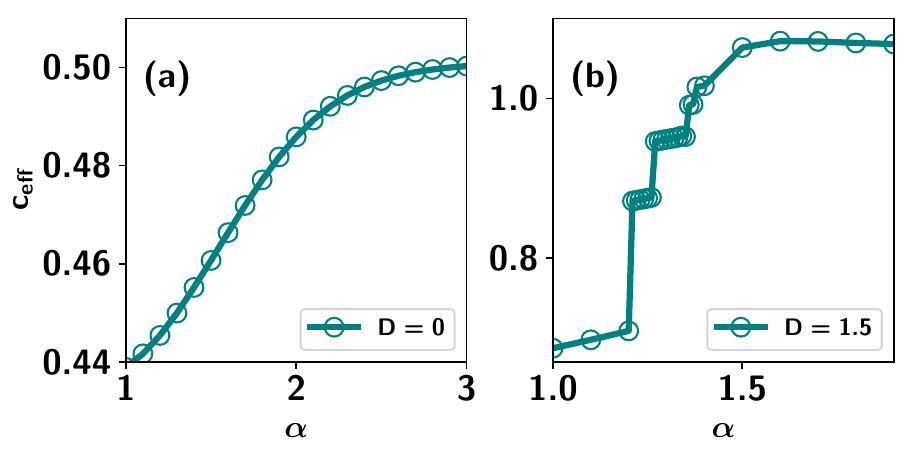}
    \caption{{\bf Effective central charge, $c_{eff}$ (ordinate) against the  the power-law exponent $\alpha$ (abscissa).} We obtain $c_{eff}$ by fitting the block entanglement entropy in the logarithmic scaling equation,  $S_{l}^C =  \frac{c_{eff}}{3}\log_2\Big[\frac{N}{\pi}\sin\frac{\pi l}{N}\Big]+a$ of the ground state at criticality, i.e., \(h_c=1.0\). (a) The non-linear but smooth increase of $c_{eff}$  with \(\alpha\) in the absence of asymmetric DM interaction. (b) The abrupt change in $c_{eff}$ against \(\alpha\) is observed in the presence of DM interactions for a fixed values of \(D\). 
    %The first kink is due to the gapped to gapless transition with the variation of \(\alpha\) while 
    The kinks possibly arises due to the presence of chirality (gapless).
    All other specifications are same as in Fig. \ref{fig:chiral_phase_alpha}. All the axes are dimensionless.}
    \label{fig:ceff}
\end{figure}

\section{decay of dynamical correlation: uncovering chiral and achiral phases}
\label{sec:dynamics}

Let us now study the trends of physical quantities including quantum correlations in the time-evolved state in the transient and the steady state regimes, starting from the product state or the ground state which evolves according to the LR Hamiltonian with DM interactions. Our aim is to demonstrate that the scaling of correlations with time or with distant sites can indicate the transition from the chiral phase to the achiral ones. 

{\color{black} Note that the NISQ quantum computers are capable of simulating the dynamics of the correlators under scrutiny in this work.  A recent study has shown that the truncated Taylor series approach may be used to simulate the dynamics of the correlators \cite{kishor_bharti_truncated_scipost} while   another article \cite{sycamore_dtc_prx_quantum} describes the realization of a discrete time crystal (DTC) on the Google Sycamore quantum processor.  Using a quantum auto-correlator defined as $\langle \sigma^z(0)\sigma^z(n) \rangle$ at the discrete time-steps, $n$ of the associated Floquet drive, the DTC's existence has been verified.  The Hamiltonian has been efficiently simulated using the ``trotterized" circuit.  Therefore, we hypothesize that ``trotterization" or truncated Taylor series expansion can be used to examine the dynamics of the correlators  in our study in a NISQ computer.
}

\subsection{Relaxation behavior of dynamical correlations}

Let us first present the method to capture the relaxation of dynamical correlation \cite{makki_prb_2022} which is shown to comprehend distinct phases in equilibrium. The initial state and the evolution of the system in the momentum basis can be represented by the Bogoliubov–de Gennes (BdG) equation, given as

\begin{equation}
    \psi_{k}(0)=[u_k(0), v_k(0)]^T\quad \text{and} \quad \psi_{k}(t)=e^{-iH_kt}\psi_{k}(0).
\end{equation}
The corresponding dynamical correlation (DC) is defined as \(C_{mn}(t)=\expval{c_m^\dagger c_n}{\psi(t)}\) which is a fermionic correlation between two modes. The relaxation of this correlation can be measured by taking the difference between correlation in arbitrary time and time in which the system reaches the steady state, given by \(\delta C_{mn}(t)=C_{mn}(t)-C_{mn}(\infty)\), describing the decay of the correlation function with time. Let the Bogoliubov angle of the initial and quenched Hamiltonian be $\eta_k$ and \(\tilde{\eta}_k\) respectively and the difference between the angels is given as \(\alpha_k=\eta_k-\tilde{\eta}_k\). The relaxation of DC takes the form \cite{cao_prb_2024}
\begin{align}
    \nonumber &\delta C_{mn}(t) = \\ \label{eq:second} & \underset{k\in \text{gapped}}{\int}\sin 2\tilde{\eta}_k\sin 2\alpha_k \cos(\epsilon_k+\epsilon_{-k})t\cos[k(n-m)]= 
    \\ \label{eq:third} &\underset{k\notin \text{gapless}}{\int}\sin 2\tilde{\eta}_k\sin 2\alpha_k \cos(\epsilon_k+\epsilon_{-k})t\cos[k(n-m)],
\end{align}
where the initial state belongs to the gapped  and gapless regions in Eqs. (\ref{eq:second}) and (\ref{eq:third}) respectively. In our case, we have chosen \(\abs{m-n}=1\), but our result holds for any arbitrary \(m \text{ and }n\).
\begin{figure}
    \centering
    \includegraphics[width=\linewidth]{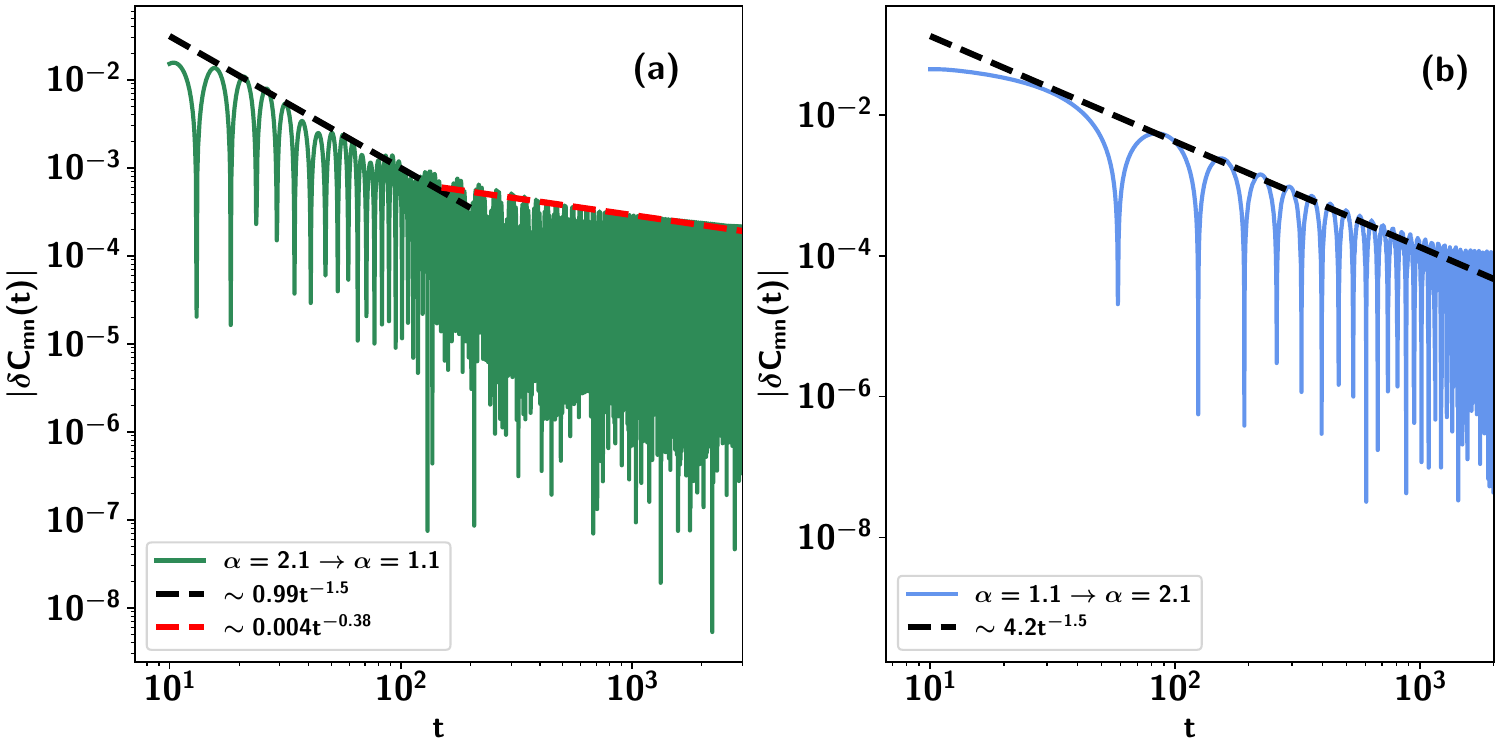}
    \caption{{\bf The transient  dynamics of the $\delta C_{mn}(t)$ (ordinate) as a function of time $t$ (abscissa) in the log-range Ising model with DM interactions}.  In both (a) and (b), \(h=-0.5, D=1.3, \gamma=1.0, \text{ and } N=30\text{K}\) for \(t=0\) and for all other values of \(t>0\). For evolution,  the sudden quench is performed from $\alpha_i = 2.1$ to $\alpha_q= 1.1$ in (a) and vice-versa in (b). A dashed line is a fit of the function indicating the overall behavior. All the axes are dimensionless.}
    \label{fig:dyn_exponent_gapless_gapped}
\end{figure}
When $\delta C_{mn}(t) \sim t^{-\chi}$, $\chi$ indicates the dynamical relaxation exponent characterized by the distinctive values of the system parameters.

In the case of the nearest-neighbor $XY$ model, it was shown that the exponent depends on the phase being commensurate or incommensurate. Specifically, in the commensurate phase, the derivative of the dispersion relation with respect to the momentum vanishes at the extreme ends of the Brillouin zone, while in the incommensurate phase, the derivative vanishes inside the Brillouin zone. If the initial state is a product state, i.e., the ground state of the model when the magnetic field goes to infinity and is then quenched with the Hamiltonian corresponding to the commensurate phase, we find $\chi = 3/2 $ while $\chi = 1/2$ when quenched with the incommensurate phase \cite{makki_prb_2022}. 

In the NN $XY$ model with DM interaction, if the initial state belongs to the gapless phase, the exponent changes to $\chi = 1$ if the quenching Hamiltonian is in the commensurate phase, which indicates the presence of DM interaction. In the case of quenching by a Hamiltonian in the incommensurate phase, the exponent is either $\chi = 1/2$ or $\chi = 1$ depending on the parameters of the initial state \cite{cao_prb_2024}. On the other hand, considering the LR Kitaev chain, it was found that the scaling law behaves similarly to the SR ones when the initial and post-quench Hamiltonian are non-critical although the exponent changes with the parameters chosen from an equilibrium phase transition \cite{makki_prb_2022}. 

The model considered in this work possesses both LR DM interactions and LR interactions in the $xy$-plane and hence the interplay between $\alpha$ and $D$ as mentioned in the static scenario can influence the scaling of $\delta C_{mn}(t)$, resulting to different $\chi$. For investigations, possible three situations arise in the non-local regime $(\alpha < 1)$ - (ni) both pre- and post-quench Hamiltonian are in the chiral phase; (nii) only pre-quench Hamiltonian belongs to the chiral phase; (niii) only post-quench Hamiltonian is chosen from the chiral phase. All three cases can also be considered when $\alpha$ is chosen from a quasi-local regime, i.e., when $\alpha > 1$ which we refer to as (qi) - (qiii). Eg., we find that $\chi \sim 1.5$ for the case of (niii) and when (qii) $\chi \sim 0.3$ (see Fig. \ref{fig:dyn_exponent_gapless_gapped}). 
% In the same spirit of recognising the phase using the dynamical relaxation exponent. As seen in the previous section, there is a paradoxical behaviour with the generation of gapless phase around the $\alpha=1$ point. In such a paradigm, we first start with the ground state from the gapped phase in the non-local regime i.e., $\alpha < 1$ and then quench by the Hamiltonian in the gapless phase by changing the magnetic field, we observe the exponent goes as $\chi=-1.4$. When $\alpha>1$ the initial state is the ground state of the Hamiltonian corresponding to the gapless phase and quench it using a Hamiltonian in the gapped phase, we observe that the exponent is given by $\chi=-0.3$. On the other hand, when the initial state is in gapped phase in \(\alpha>1\) region and quench is occurred in the gapless region, in that case, there is a change over of exponents, \(\chi\) from \(-1\) to \(-0.1\) which is appeared due to long-range interaction. 
This shows that the exponent in DC can distinguish between the gapped and the gapless phases  of the post- and pre-quench Hamiltonians  (see Fig. \ref{fig:dyn_exponent_gapless_gapped}). In addition, the robustness of correlation can be defined as a quantity $\frac{1}{\chi}$ such that if the value of $\chi$ is small, the system does not reach the steady state correlation value for a significantly long time. This reinforces the above result that the state from the gapless phase contains a strong correlation both in terms of space and time. 
\begin{figure*}
    \centering
    \includegraphics[width=\textwidth]{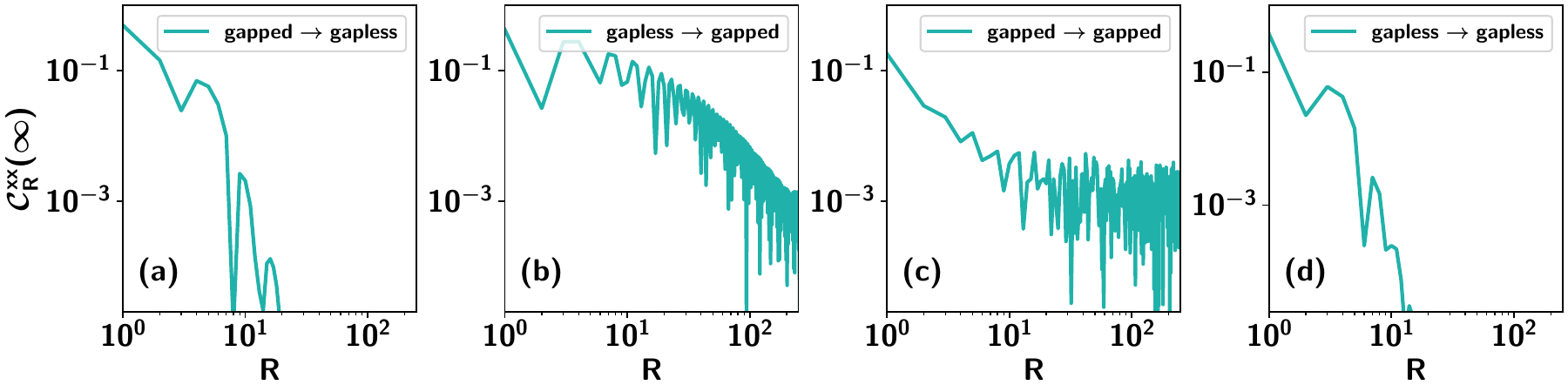}
    \includegraphics[width=\textwidth]{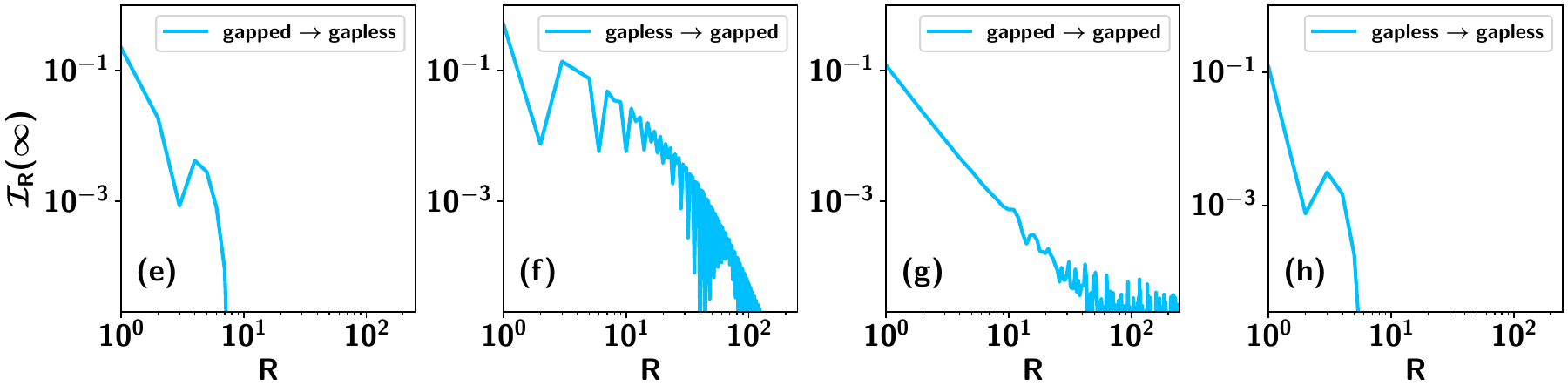}
    \caption{{\bf Classical and quantum mutual information in the steady state as a function of distance between the spins, $R$ (abscissa) of the LR Ising model with \(D=1.3\).} (a)-(d) Mutual information or total correlation $\mathcal{I}_R (\infty)$ (ordinate) in four different combinations of initial and quenching Hamiltonian picked from the gapped and gapless phases as mentioned in Fig. \ref{fig:dyn_exponent_gapless_gapped}. A similar plot for classical correlation $\mathcal{C}_R^{xx} (\infty)$ in (e)-(f). The strength of the magnetic field $h = -0.5$. Here \(N = 512\). Note the striking similarity between the top and bottom rows. All the axes are dimensionless. }
    \label{fig:mutual_gapped_gappless}
\end{figure*}
\subsection{Slow decay of total correlation and two-point correlation with LR DM interaction}
\label{sec:steady_state_corre}

We now focus on six scenarios ((ni) - (niii) and (qi) - (qiii)) to study the decay pattern of total correlation and classical correlation with $R$ in a steady state limit, represented by $\mathcal{I}_R(\infty)$ and $\mathcal{C}^{xx}_R(\infty)$ respectively. Even though the evolution is unitary, as we are studying the properties of a subsystem, a steady state can be reached in the dynamics. When the quenching Hamiltonian is in the gapped phase ((nii)), the steady state maintains both classical and quantum correlations between distant spins, and  persists even at large distance $R$ (see Figs. \ref{fig:mutual_gapped_gappless}(b), (c), (f) and (g)). Conversely, a gapless quenching Hamiltonian ((ni) and (niii)) results in a steady state where classical correlations are present only between nearby spins as depicted in Figs. \ref{fig:mutual_gapped_gappless} (a), (d), (e) and (h). Notice that the decaying natures of $\mathcal{I_R}(\infty)$ and $\mathcal{C}^{xx}_R(\infty)$ are independent of the initial state. Remarkably, the behavior of total and classical correlations remains similar when the distance between the spins is small.

\subsection{Effects of long-range in dynamics of EE}
\label{sec:dynamics_entropy}

In a final attempt to distinguish the gapless and gapped phases under dynamics, we investigate the
growth of entropy with time. The initial states are prepared as the ground states of gapped as well as from the gapless regions and EE, $S_l(\rho(t)) \equiv S_t$ for a fixed $l$ is computed after quenching the system to the gapped or gapless regions by changing the parameters $\alpha$ for a given $\gamma, h$ and $D$. Let us denote the initial and final fall-off rates be $\alpha_i$ and $\alpha_q$ respectively. Firstly, we notice that the rate of growth is shown to be linear in time in the case of the $XY$ model with NN \cite{XY_entropy_growth_calabrese_2008} and long-range interacting model \cite{Schachenmayer_entropy_growth_2013,regemortel_pra_2016,buyskikh_pra_2016,Maity2019_review}, with the initial state being a product state. The trends of $S_t$ with time depend on both $\alpha_i$ and $\alpha_q$, thereby indicating its dependence on pre- and post-quench Hamiltonian. 
\begin{figure*}
    \centering
    \includegraphics[width=\textwidth]{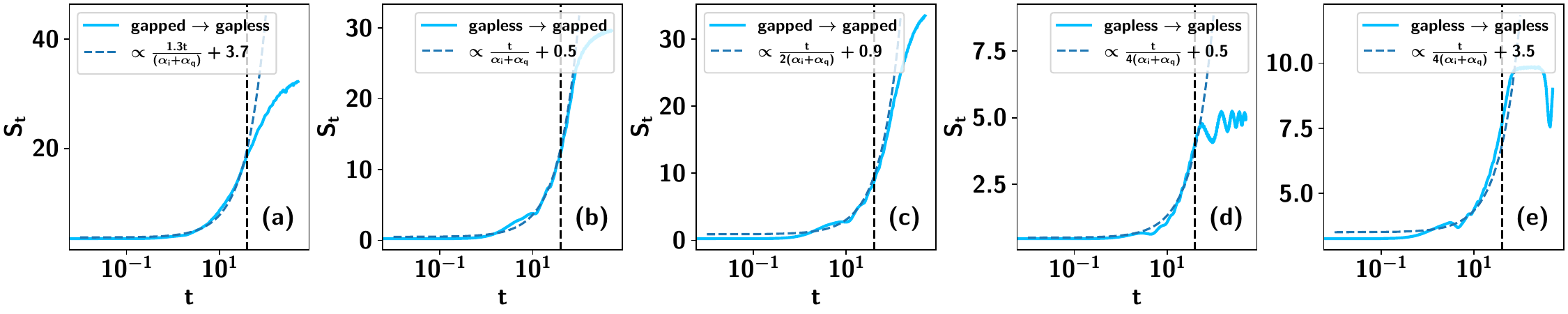}
    \caption{{\bf Entanglement entropy of block-size $l$ of the quenched state $S_t$ (ordinate) with respect to time $t$ (abscissa).} The parameter sets are given as follows:  (a)  \(h=-0.5, D=1.3, \alpha_i=1.1, \alpha_q=2.1\) that indicates a quenching from a gapped to a gapless phase, (b) \(h= -0.5, D=1.3, \alpha_i=2.1, \alpha_q=1.1\) -- from a gapless to a gapped quench, (c) \(h=-0.5, D=1.3, \alpha_i=1.1, \alpha_q=1.2\) - quenching from a gapped to a gapped phase, (d) \(h=-0.5, D=2.5, \alpha_i=0.5, \alpha_q=2.5\) -- quenching from a gapless to a gapless phase, and (e) \(h=-0.5, D=2.5, \alpha_i=2.5, \alpha_q=0.5\)-- from a gapless to a gapless quenching. In (d) and (e), a vertical dashed line at $t = \frac{l}{2} $  indicates the saturation time while in (a), (b), (c) saturation does not happen at \(t=\frac{l}{2}\) where \(l=80\). In (a), (b), (d), and (e),  we present a universal law for the growth of entanglement entropy $S_t^\pm$ using a dashed line, mentioned in Eq. (\ref{eq:dyn_scalingEE}). All the axes are dimensionless.  }
    \label{fig:entropy_dynamics}
\end{figure*}
  Secondly, the rate of increment of EE with time is faster if the initial state is in a gapless region and the EE value gets saturated around time \( t \sim l/2\), as observed in the case of the nearest-neighbor Ising model quenched from a non-critical regime to a critical one \cite{evolution_entropy_calabrese}. However, if the initial and final states are chosen from a gapless region, the increment rate is slow and the saturation is reached faster around \(t=l/2\) compared to a gapped case (see Fig. \ref{fig:entropy_dynamics}(d) and (e)). Also, the scaling of entropy increment until it reaches saturation can be approximated as
\begin{equation}
    S_t^{\pm} \propto a_1\frac{t}{\abs{\alpha_i \pm \alpha_q}}+a_2,
    \label{eq:dyn_scalingEE}
\end{equation}
where \(a_1\), and \(a_2\) are constants (see Fig. \ref{fig:entropy_dynamics}).
%From the gapless to a gapped quenching, the rate of increase of $S_t$ is slow which depicts the slow correlation spreading. Interestingly, 
When both the initial and the final Hamiltonian are gapped or one of the pre- or post-quench Hamiltonian is gapped, the saturation of EE requires much larger time than that obtained with  the Hamiltonians corresponding to the initial and final states being gapless.

\section{conclusion}
\label{sec:conclusion}

Quantum long-range systems are intriguing due to their non-local properties, often displaying counter-intuitive phenomena compared to a short-range Hamiltonian. These models are pervasive in physically realizable systems such as trapped ions and atomic, molecular, and optical setups.   Therefore, comprehending their universal features is essential for the advancement of quantum technologies and condensed matter physics.
%, offering insights into diverse physical phenomena such as transitions, correlations, and entropy dynamics are crucial for advancing fields like condensed matter physics, quantum information theory and field theory.

%We investigated the influence of long-range Dzyaloshinskii-Moriya (DM) interactions in the extended Ising model. 

We found that the long-range extended Dzyaloshinskii-Moriya (DM) interactions with the extended $XY$ model have unique consequences on the critical behavior of the system with a variation of the fall-off rate, \(\alpha\), in which the range of interactions between the sites decays. Specifically,  even with a comparable or more DM interactions than the anisotropy parameter, we proved that the system remains gapped depending on \(\alpha\), which  varies with the direction of the transverse magnetic field although the gapless phase still possesses chiral order.  
By analyzing classical and total correlations between two arbitrary sites of the ground state in the presence of DM interactions, we observed that both  of them decay with the distance between the sites and the decay exponent is independent of \(\alpha\) in the gapless region while it depends on the system parameters in the gapped regime and  decays faster in this case than the gapless zone.  Additionally, we exhibited   the departures from the universal trend in the block entanglement entropy in presence of long-range DM interactions which could successfully capture the interplay between long-range and DM interactions.

%Interestingly, the gapless and the chiral phases exhibit proportional relationships with respect to each other and a long-range current flow due to \(\alpha\) is established. 
% decay as \(R^{-1}\) in the gapless region independent of  \(\alpha\), while classical correlation \(x\)-direction shows a consistent decay pattern with increasing qubit separation. 
%at high long-range strengths, where the conformal properties break down notably at small \(\alpha\). These findings contribute to understanding the intricate interplay of long-range interactions and DM effects in quantum systems.

This study investigated how dynamical  correlations change in time when the energy spectra of the initial and the evolving Hamiltonian are gapped and gapless, driven by  \(\alpha\) for a fixed strengths of DM interactions, and magnetic fields. 
%Importantly, unlike other studies in the literature, the quenching takes place by varying \(\alpha\). 
The scaling of both classical correlations and mutual information in the  steady states reveals that the sharing of correlations is favorable when the evolving Hamiltonian is gapped.  Moreover, we found that the growth of entanglement entropy with time  is faster in the gapless regime than in the gapped phase and it saturates quickly when both the pre- and post-quench Hamiltonians are gapless.   On the one hand, these findings contribute to the understanding of fundamental aspects of quantum phase transitions, while on the other, sharability of correlations and entanglement  of the ground state as well as their dynamical patterns can be vital for establishing quantum links amongst quantum computers. 

\acknowledgments

We acknowledge the support from the Interdisciplinary Cyber Physical Systems (ICPS) program of the Department of Science and Technology (DST), India, Grant No.: DST/ICPS/QuST/Theme- 1/2019/23. We  acknowledge the use of \href{https://github.com/titaschanda/QIClib}{QIClib} -- a modern C++ library for general purpose quantum information processing and quantum computing (\url{https://titaschanda.github.io/QIClib}), \textcolor{black}{numerical computation of Pfaffians by M. Wimmer \cite{Wimmer2012Aug}} and the cluster computing facility at the Harish-Chandra Research Institute. This research was supported in part by the ``INFOSYS" scholarship for senior students. Funded by the European Union. Views and opinions expressed are however those of the author(s) only and do not necessarily reflect those of the European Union or the European Commission. Neither the European Union nor the granting authority can be held responsible for them.
This project has received funding from the European Union’s Horizon Europe research and innovation programme under grant agreement No 101080086 NeQST.

% \showsummary
% \bibliographystyle{apsrev4-1}
\bibliography{bibfile.bib}

\begin{thebibliography}{91}%
\makeatletter
\providecommand \@ifxundefined [1]{%
 \@ifx{#1\undefined}
}%
\providecommand \@ifnum [1]{%
 \ifnum #1\expandafter \@firstoftwo
 \else \expandafter \@secondoftwo
 \fi
}%
\providecommand \@ifx [1]{%
 \ifx #1\expandafter \@firstoftwo
 \else \expandafter \@secondoftwo
 \fi
}%
\providecommand \natexlab [1]{#1}%
\providecommand \enquote  [1]{``#1''}%
\providecommand \bibnamefont  [1]{#1}%
\providecommand \bibfnamefont [1]{#1}%
\providecommand \citenamefont [1]{#1}%
\providecommand \href@noop [0]{\@secondoftwo}%
\providecommand \href [0]{\begingroup \@sanitize@url \@href}%
\providecommand \@href[1]{\@@startlink{#1}\@@href}%
\providecommand \@@href[1]{\endgroup#1\@@endlink}%
\providecommand \@sanitize@url [0]{\catcode `\\12\catcode `\$12\catcode
  `\&12\catcode `\#12\catcode `\^12\catcode `\_12\catcode `\%12\relax}%
\providecommand \@@startlink[1]{}%
\providecommand \@@endlink[0]{}%
\providecommand \url  [0]{\begingroup\@sanitize@url \@url }%
\providecommand \@url [1]{\endgroup\@href {#1}{\urlprefix }}%
\providecommand \urlprefix  [0]{URL }%
\providecommand \Eprint [0]{\href }%
\providecommand \doibase [0]{https://doi.org/}%
\providecommand \selectlanguage [0]{\@gobble}%
\providecommand \bibinfo  [0]{\@secondoftwo}%
\providecommand \bibfield  [0]{\@secondoftwo}%
\providecommand \translation [1]{[#1]}%
\providecommand \BibitemOpen [0]{}%
\providecommand \bibitemStop [0]{}%
\providecommand \bibitemNoStop [0]{.\EOS\space}%
\providecommand \EOS [0]{\spacefactor3000\relax}%
\providecommand \BibitemShut  [1]{\csname bibitem#1\endcsname}%
\let\auto@bib@innerbib\@empty
%</preamble>
\bibitem [{\citenamefont {Diessel}\ \emph {et~al.}(2023)\citenamefont
  {Diessel}, \citenamefont {Diehl}, \citenamefont {Defenu}, \citenamefont
  {Rosch},\ and\ \citenamefont {Chiocchetta}}]{diessel_prr_2023}%
  \BibitemOpen
  \bibfield  {author} {\bibinfo {author} {\bibfnamefont {O.~K.}\ \bibnamefont
  {Diessel}}, \bibinfo {author} {\bibfnamefont {S.}~\bibnamefont {Diehl}},
  \bibinfo {author} {\bibfnamefont {N.}~\bibnamefont {Defenu}}, \bibinfo
  {author} {\bibfnamefont {A.}~\bibnamefont {Rosch}},\ and\ \bibinfo {author}
  {\bibfnamefont {A.}~\bibnamefont {Chiocchetta}},\ }\bibfield  {title}
  {\bibinfo {title} {Generalized higgs mechanism in long-range-interacting
  quantum systems},\ }\href {https://doi.org/10.1103/PhysRevResearch.5.033038}
  {\bibfield  {journal} {\bibinfo  {journal} {Phys. Rev. Res.}\ }\textbf
  {\bibinfo {volume} {5}},\ \bibinfo {pages} {033038} (\bibinfo {year}
  {2023})}\BibitemShut {NoStop}%
\bibitem [{\citenamefont {Defenu}\ \emph
  {et~al.}(2024{\natexlab{a}})\citenamefont {Defenu}, \citenamefont {Lerose},\
  and\ \citenamefont {Pappalardi}}]{Defenu2024Jun}%
  \BibitemOpen
  \bibfield  {author} {\bibinfo {author} {\bibfnamefont {N.}~\bibnamefont
  {Defenu}}, \bibinfo {author} {\bibfnamefont {A.}~\bibnamefont {Lerose}},\
  and\ \bibinfo {author} {\bibfnamefont {S.}~\bibnamefont {Pappalardi}},\
  }\bibfield  {title} {\bibinfo {title} {{Out-of-equilibrium dynamics of
  quantum many-body systems with long-range interactions}},\ }\href
  {https://doi.org/10.1016/j.physrep.2024.04.005} {\bibfield  {journal}
  {\bibinfo  {journal} {Phys. Rep.}\ }\textbf {\bibinfo {volume} {1074}},\
  \bibinfo {pages} {1} (\bibinfo {year} {2024}{\natexlab{a}})}\BibitemShut
  {NoStop}%
\bibitem [{\citenamefont {Defenu}\ \emph {et~al.}(2018)\citenamefont {Defenu},
  \citenamefont {Enss}, \citenamefont {Kastner},\ and\ \citenamefont
  {Morigi}}]{defenu_prl_2018}%
  \BibitemOpen
  \bibfield  {author} {\bibinfo {author} {\bibfnamefont {N.}~\bibnamefont
  {Defenu}}, \bibinfo {author} {\bibfnamefont {T.}~\bibnamefont {Enss}},
  \bibinfo {author} {\bibfnamefont {M.}~\bibnamefont {Kastner}},\ and\ \bibinfo
  {author} {\bibfnamefont {G.}~\bibnamefont {Morigi}},\ }\bibfield  {title}
  {\bibinfo {title} {Dynamical critical scaling of long-range interacting
  quantum magnets},\ }\href {https://doi.org/10.1103/PhysRevLett.121.240403}
  {\bibfield  {journal} {\bibinfo  {journal} {Phys. Rev. Lett.}\ }\textbf
  {\bibinfo {volume} {121}},\ \bibinfo {pages} {240403} (\bibinfo {year}
  {2018})}\BibitemShut {NoStop}%
\bibitem [{\citenamefont {Defenu}\ \emph {et~al.}(2019)\citenamefont {Defenu},
  \citenamefont {Morigi}, \citenamefont {Dell'Anna},\ and\ \citenamefont
  {Enss}}]{defenu_prb_2019}%
  \BibitemOpen
  \bibfield  {author} {\bibinfo {author} {\bibfnamefont {N.}~\bibnamefont
  {Defenu}}, \bibinfo {author} {\bibfnamefont {G.}~\bibnamefont {Morigi}},
  \bibinfo {author} {\bibfnamefont {L.}~\bibnamefont {Dell'Anna}},\ and\
  \bibinfo {author} {\bibfnamefont {T.}~\bibnamefont {Enss}},\ }\bibfield
  {title} {\bibinfo {title} {Universal dynamical scaling of long-range
  topological superconductors},\ }\href
  {https://doi.org/10.1103/PhysRevB.100.184306} {\bibfield  {journal} {\bibinfo
   {journal} {Phys. Rev. B}\ }\textbf {\bibinfo {volume} {100}},\ \bibinfo
  {pages} {184306} (\bibinfo {year} {2019})}\BibitemShut {NoStop}%
\bibitem [{\citenamefont {Solfanelli}\ and\ \citenamefont
  {Defenu}(2024)}]{solfanelli_arxiv_2024}%
  \BibitemOpen
  \bibfield  {author} {\bibinfo {author} {\bibfnamefont {A.}~\bibnamefont
  {Solfanelli}}\ and\ \bibinfo {author} {\bibfnamefont {N.}~\bibnamefont
  {Defenu}},\ }\href {https://arxiv.org/abs/2406.14651} {\bibinfo {title}
  {Universality in long-range interacting systems: the effective dimension
  approach}} (\bibinfo {year} {2024}),\ \Eprint
  {https://arxiv.org/abs/2406.14651} {arXiv:2406.14651 [cond-mat.stat-mech]}
  \BibitemShut {NoStop}%
\bibitem [{\citenamefont {Defenu}\ \emph
  {et~al.}(2024{\natexlab{b}})\citenamefont {Defenu}, \citenamefont {Mukamel},\
  and\ \citenamefont {Ruffo}}]{defenu_arxiv_2024}%
  \BibitemOpen
  \bibfield  {author} {\bibinfo {author} {\bibfnamefont {N.}~\bibnamefont
  {Defenu}}, \bibinfo {author} {\bibfnamefont {D.}~\bibnamefont {Mukamel}},\
  and\ \bibinfo {author} {\bibfnamefont {S.}~\bibnamefont {Ruffo}},\ }\href
  {https://arxiv.org/abs/2403.06673} {\bibinfo {title} {Ensemble inequivalence
  in long-range quantum systems}} (\bibinfo {year} {2024}{\natexlab{b}}),\
  \Eprint {https://arxiv.org/abs/2403.06673} {arXiv:2403.06673
  [cond-mat.stat-mech]} \BibitemShut {NoStop}%
\bibitem [{\citenamefont {Saffman}\ \emph {et~al.}(2010)\citenamefont
  {Saffman}, \citenamefont {Walker},\ and\ \citenamefont
  {M\o{}lmer}}]{rydberg_review_experiments}%
  \BibitemOpen
  \bibfield  {author} {\bibinfo {author} {\bibfnamefont {M.}~\bibnamefont
  {Saffman}}, \bibinfo {author} {\bibfnamefont {T.~G.}\ \bibnamefont
  {Walker}},\ and\ \bibinfo {author} {\bibfnamefont {K.}~\bibnamefont
  {M\o{}lmer}},\ }\bibfield  {title} {\bibinfo {title} {Quantum information
  with rydberg atoms},\ }\href {https://doi.org/10.1103/RevModPhys.82.2313}
  {\bibfield  {journal} {\bibinfo  {journal} {Rev. Mod. Phys.}\ }\textbf
  {\bibinfo {volume} {82}},\ \bibinfo {pages} {2313} (\bibinfo {year}
  {2010})}\BibitemShut {NoStop}%
\bibitem [{\citenamefont {Lahaye}\ \emph {et~al.}(2009)\citenamefont {Lahaye},
  \citenamefont {Menotti}, \citenamefont {Santos}, \citenamefont {Lewenstein},\
  and\ \citenamefont {Pfau}}]{dipolar_longrange}%
  \BibitemOpen
  \bibfield  {author} {\bibinfo {author} {\bibfnamefont {T.}~\bibnamefont
  {Lahaye}}, \bibinfo {author} {\bibfnamefont {C.}~\bibnamefont {Menotti}},
  \bibinfo {author} {\bibfnamefont {L.}~\bibnamefont {Santos}}, \bibinfo
  {author} {\bibfnamefont {M.}~\bibnamefont {Lewenstein}},\ and\ \bibinfo
  {author} {\bibfnamefont {T.}~\bibnamefont {Pfau}},\ }\bibfield  {title}
  {\bibinfo {title} {The physics of dipolar bosonic quantum gases},\ }\href
  {https://doi.org/10.1088/0034-4885/72/12/126401} {\bibfield  {journal}
  {\bibinfo  {journal} {Reports on Progress in Physics}\ }\textbf {\bibinfo
  {volume} {72}},\ \bibinfo {pages} {126401} (\bibinfo {year}
  {2009})}\BibitemShut {NoStop}%
\bibitem [{\citenamefont {Carr}\ \emph {et~al.}(2009)\citenamefont {Carr},
  \citenamefont {DeMille}, \citenamefont {Krems},\ and\ \citenamefont
  {Ye}}]{cold_gas_long_range_review}%
  \BibitemOpen
  \bibfield  {author} {\bibinfo {author} {\bibfnamefont {L.~D.}\ \bibnamefont
  {Carr}}, \bibinfo {author} {\bibfnamefont {D.}~\bibnamefont {DeMille}},
  \bibinfo {author} {\bibfnamefont {R.~V.}\ \bibnamefont {Krems}},\ and\
  \bibinfo {author} {\bibfnamefont {J.}~\bibnamefont {Ye}},\ }\bibfield
  {title} {\bibinfo {title} {Cold and ultracold molecules: science, technology
  and applications},\ }\href {https://doi.org/10.1088/1367-2630/11/5/055049}
  {\bibfield  {journal} {\bibinfo  {journal} {New Journal of Physics}\ }\textbf
  {\bibinfo {volume} {11}},\ \bibinfo {pages} {055049} (\bibinfo {year}
  {2009})}\BibitemShut {NoStop}%
\bibitem [{\citenamefont {Blatt}\ and\ \citenamefont
  {Roos}(2012)}]{trapped_ion_1}%
  \BibitemOpen
  \bibfield  {author} {\bibinfo {author} {\bibfnamefont {R.}~\bibnamefont
  {Blatt}}\ and\ \bibinfo {author} {\bibfnamefont {C.~F.}\ \bibnamefont
  {Roos}},\ }\bibfield  {title} {\bibinfo {title} {{Quantum simulations with
  trapped ions}},\ }\href {https://doi.org/10.1038/nphys2252} {\bibfield
  {journal} {\bibinfo  {journal} {Nat. Phys.}\ }\textbf {\bibinfo {volume}
  {8}},\ \bibinfo {pages} {277} (\bibinfo {year} {2012})}\BibitemShut {NoStop}%
\bibitem [{\citenamefont {Blatt}\ and\ \citenamefont
  {Wineland}(2008)}]{trapped_ion_2}%
  \BibitemOpen
  \bibfield  {author} {\bibinfo {author} {\bibfnamefont {R.}~\bibnamefont
  {Blatt}}\ and\ \bibinfo {author} {\bibfnamefont {D.}~\bibnamefont
  {Wineland}},\ }\bibfield  {title} {\bibinfo {title} {{Entangled states of
  trapped atomic ions}},\ }\href {https://doi.org/10.1038/nature07125}
  {\bibfield  {journal} {\bibinfo  {journal} {Nature}\ }\textbf {\bibinfo
  {volume} {453}},\ \bibinfo {pages} {1008} (\bibinfo {year}
  {2008})}\BibitemShut {NoStop}%
\bibitem [{\citenamefont {Schneider}\ \emph {et~al.}(2012)\citenamefont
  {Schneider}, \citenamefont {Porras},\ and\ \citenamefont
  {Schaetz}}]{trapped_ion_3}%
  \BibitemOpen
  \bibfield  {author} {\bibinfo {author} {\bibfnamefont {C.}~\bibnamefont
  {Schneider}}, \bibinfo {author} {\bibfnamefont {D.}~\bibnamefont {Porras}},\
  and\ \bibinfo {author} {\bibfnamefont {T.}~\bibnamefont {Schaetz}},\
  }\bibfield  {title} {\bibinfo {title} {Experimental quantum simulations of
  many-body physics with trapped ions},\ }\href
  {https://doi.org/10.1088/0034-4885/75/2/024401} {\bibfield  {journal}
  {\bibinfo  {journal} {Reports on Progress in Physics}\ }\textbf {\bibinfo
  {volume} {75}},\ \bibinfo {pages} {024401} (\bibinfo {year}
  {2012})}\BibitemShut {NoStop}%
\bibitem [{\citenamefont {Ritsch}\ \emph {et~al.}(2013)\citenamefont {Ritsch},
  \citenamefont {Domokos}, \citenamefont {Brennecke},\ and\ \citenamefont
  {Esslinger}}]{cold_atom_cavity_long_range_1}%
  \BibitemOpen
  \bibfield  {author} {\bibinfo {author} {\bibfnamefont {H.}~\bibnamefont
  {Ritsch}}, \bibinfo {author} {\bibfnamefont {P.}~\bibnamefont {Domokos}},
  \bibinfo {author} {\bibfnamefont {F.}~\bibnamefont {Brennecke}},\ and\
  \bibinfo {author} {\bibfnamefont {T.}~\bibnamefont {Esslinger}},\ }\bibfield
  {title} {\bibinfo {title} {Cold atoms in cavity-generated dynamical optical
  potentials},\ }\href {https://doi.org/10.1103/RevModPhys.85.553} {\bibfield
  {journal} {\bibinfo  {journal} {Rev. Mod. Phys.}\ }\textbf {\bibinfo {volume}
  {85}},\ \bibinfo {pages} {553} (\bibinfo {year} {2013})}\BibitemShut
  {NoStop}%
\bibitem [{\citenamefont {Farokh~Mivehvar}\ and\ \citenamefont
  {Ritsch}(2021)}]{cold_atom_cavity_long_range_2}%
  \BibitemOpen
  \bibfield  {author} {\bibinfo {author} {\bibfnamefont {T.~D.}\ \bibnamefont
  {Farokh~Mivehvar}, \bibfnamefont {Francesco~Piazza}}\ and\ \bibinfo {author}
  {\bibfnamefont {H.}~\bibnamefont {Ritsch}},\ }\bibfield  {title} {\bibinfo
  {title} {Cavity qed with quantum gases: new paradigms in many-body physics},\
  }\href {https://doi.org/10.1080/00018732.2021.1969727} {\bibfield  {journal}
  {\bibinfo  {journal} {Advances in Physics}\ }\textbf {\bibinfo {volume}
  {70}},\ \bibinfo {pages} {1} (\bibinfo {year} {2021})},\ \Eprint
  {https://arxiv.org/abs/https://doi.org/10.1080/00018732.2021.1969727}
  {https://doi.org/10.1080/00018732.2021.1969727} \BibitemShut {NoStop}%
\bibitem [{\citenamefont {Maghrebi}\ \emph {et~al.}(2016)\citenamefont
  {Maghrebi}, \citenamefont {Gong}, \citenamefont {Foss-Feig},\ and\
  \citenamefont {Gorshkov}}]{PhysRevB.93.125128}%
  \BibitemOpen
  \bibfield  {author} {\bibinfo {author} {\bibfnamefont {M.~F.}\ \bibnamefont
  {Maghrebi}}, \bibinfo {author} {\bibfnamefont {Z.-X.}\ \bibnamefont {Gong}},
  \bibinfo {author} {\bibfnamefont {M.}~\bibnamefont {Foss-Feig}},\ and\
  \bibinfo {author} {\bibfnamefont {A.~V.}\ \bibnamefont {Gorshkov}},\
  }\bibfield  {title} {\bibinfo {title} {Causality and quantum criticality in
  long-range lattice models},\ }\href
  {https://doi.org/10.1103/PhysRevB.93.125128} {\bibfield  {journal} {\bibinfo
  {journal} {Phys. Rev. B}\ }\textbf {\bibinfo {volume} {93}},\ \bibinfo
  {pages} {125128} (\bibinfo {year} {2016})}\BibitemShut {NoStop}%
\bibitem [{\citenamefont {Gong}\ \emph {et~al.}(2017)\citenamefont {Gong},
  \citenamefont {Foss-Feig}, \citenamefont {Brand\~ao},\ and\ \citenamefont
  {Gorshkov}}]{area_law_longrange}%
  \BibitemOpen
  \bibfield  {author} {\bibinfo {author} {\bibfnamefont {Z.-X.}\ \bibnamefont
  {Gong}}, \bibinfo {author} {\bibfnamefont {M.}~\bibnamefont {Foss-Feig}},
  \bibinfo {author} {\bibfnamefont {F.~G. S.~L.}\ \bibnamefont {Brand\~ao}},\
  and\ \bibinfo {author} {\bibfnamefont {A.~V.}\ \bibnamefont {Gorshkov}},\
  }\bibfield  {title} {\bibinfo {title} {Entanglement area laws for long-range
  interacting systems},\ }\href
  {https://doi.org/10.1103/PhysRevLett.119.050501} {\bibfield  {journal}
  {\bibinfo  {journal} {Phys. Rev. Lett.}\ }\textbf {\bibinfo {volume} {119}},\
  \bibinfo {pages} {050501} (\bibinfo {year} {2017})}\BibitemShut {NoStop}%
\bibitem [{\citenamefont {Gong}\ \emph {et~al.}(2016)\citenamefont {Gong},
  \citenamefont {Maghrebi}, \citenamefont {Hu}, \citenamefont {Foss-Feig},
  \citenamefont {Richerme}, \citenamefont {Monroe},\ and\ \citenamefont
  {Gorshkov}}]{phases_longrange}%
  \BibitemOpen
  \bibfield  {author} {\bibinfo {author} {\bibfnamefont {Z.-X.}\ \bibnamefont
  {Gong}}, \bibinfo {author} {\bibfnamefont {M.~F.}\ \bibnamefont {Maghrebi}},
  \bibinfo {author} {\bibfnamefont {A.}~\bibnamefont {Hu}}, \bibinfo {author}
  {\bibfnamefont {M.}~\bibnamefont {Foss-Feig}}, \bibinfo {author}
  {\bibfnamefont {P.}~\bibnamefont {Richerme}}, \bibinfo {author}
  {\bibfnamefont {C.}~\bibnamefont {Monroe}},\ and\ \bibinfo {author}
  {\bibfnamefont {A.~V.}\ \bibnamefont {Gorshkov}},\ }\bibfield  {title}
  {\bibinfo {title} {Kaleidoscope of quantum phases in a long-range interacting
  spin-1 chain},\ }\href {https://doi.org/10.1103/PhysRevB.93.205115}
  {\bibfield  {journal} {\bibinfo  {journal} {Phys. Rev. B}\ }\textbf {\bibinfo
  {volume} {93}},\ \bibinfo {pages} {205115} (\bibinfo {year}
  {2016})}\BibitemShut {NoStop}%
\bibitem [{\citenamefont {Uhrich}\ \emph {et~al.}(2020)\citenamefont {Uhrich},
  \citenamefont {Defenu}, \citenamefont {Jafari},\ and\ \citenamefont
  {Halimeh}}]{uhrich_prb_2020}%
  \BibitemOpen
  \bibfield  {author} {\bibinfo {author} {\bibfnamefont {P.}~\bibnamefont
  {Uhrich}}, \bibinfo {author} {\bibfnamefont {N.}~\bibnamefont {Defenu}},
  \bibinfo {author} {\bibfnamefont {R.}~\bibnamefont {Jafari}},\ and\ \bibinfo
  {author} {\bibfnamefont {J.~C.}\ \bibnamefont {Halimeh}},\ }\bibfield
  {title} {\bibinfo {title} {Out-of-equilibrium phase diagram of long-range
  superconductors},\ }\href {https://doi.org/10.1103/PhysRevB.101.245148}
  {\bibfield  {journal} {\bibinfo  {journal} {Phys. Rev. B}\ }\textbf {\bibinfo
  {volume} {101}},\ \bibinfo {pages} {245148} (\bibinfo {year}
  {2020})}\BibitemShut {NoStop}%
\bibitem [{\citenamefont {Monika}\ \emph {et~al.}(2023)\citenamefont {Monika},
  \citenamefont {Lakkaraju}, \citenamefont {Ghosh},\ and\ \citenamefont
  {De}}]{monika2023bettersensingvariablerangeinteractions}%
  \BibitemOpen
  \bibfield  {author} {\bibinfo {author} {\bibnamefont {Monika}}, \bibinfo
  {author} {\bibfnamefont {L.~G.~C.}\ \bibnamefont {Lakkaraju}}, \bibinfo
  {author} {\bibfnamefont {S.}~\bibnamefont {Ghosh}},\ and\ \bibinfo {author}
  {\bibfnamefont {A.~S.}\ \bibnamefont {De}},\ }\href
  {https://arxiv.org/abs/2307.06901} {\bibinfo {title} {Better sensing with
  variable-range interactions}} (\bibinfo {year} {2023}),\ \Eprint
  {https://arxiv.org/abs/2307.06901} {arXiv:2307.06901 [quant-ph]} \BibitemShut
  {NoStop}%
\bibitem [{\citenamefont {Ghosh}\ \emph {et~al.}(2023)\citenamefont {Ghosh},
  \citenamefont {Agarwal}, \citenamefont {Halder},\ and\ \citenamefont
  {De}}]{ghosh2023entanglementweightedgraphsuncovers}%
  \BibitemOpen
  \bibfield  {author} {\bibinfo {author} {\bibfnamefont {D.}~\bibnamefont
  {Ghosh}}, \bibinfo {author} {\bibfnamefont {K.~D.}\ \bibnamefont {Agarwal}},
  \bibinfo {author} {\bibfnamefont {P.}~\bibnamefont {Halder}},\ and\ \bibinfo
  {author} {\bibfnamefont {A.~S.}\ \bibnamefont {De}},\ }\href
  {https://arxiv.org/abs/2307.11739} {\bibinfo {title} {Entanglement of
  weighted graphs uncovers transitions in variable-range interacting models}}
  (\bibinfo {year} {2023}),\ \Eprint {https://arxiv.org/abs/2307.11739}
  {arXiv:2307.11739 [quant-ph]} \BibitemShut {NoStop}%
\bibitem [{\citenamefont {Dzyaloshinsky}(1958)}]{DZYALOSHINSKY1958241}%
  \BibitemOpen
  \bibfield  {author} {\bibinfo {author} {\bibfnamefont {I.}~\bibnamefont
  {Dzyaloshinsky}},\ }\bibfield  {title} {\bibinfo {title} {A thermodynamic
  theory of “weak” ferromagnetism of antiferromagnetics},\ }\href
  {https://doi.org/https://doi.org/10.1016/0022-3697(58)90076-3} {\bibfield
  {journal} {\bibinfo  {journal} {Journal of Physics and Chemistry of Solids}\
  }\textbf {\bibinfo {volume} {4}},\ \bibinfo {pages} {241} (\bibinfo {year}
  {1958})}\BibitemShut {NoStop}%
\bibitem [{\citenamefont {Moriya}(1960{\natexlab{a}})}]{moriya_1960}%
  \BibitemOpen
  \bibfield  {author} {\bibinfo {author} {\bibfnamefont {T.}~\bibnamefont
  {Moriya}},\ }\bibfield  {title} {\bibinfo {title} {Anisotropic superexchange
  interaction and weak ferromagnetism},\ }\href
  {https://doi.org/10.1103/PhysRev.120.91} {\bibfield  {journal} {\bibinfo
  {journal} {Phys. Rev.}\ }\textbf {\bibinfo {volume} {120}},\ \bibinfo {pages}
  {91} (\bibinfo {year} {1960}{\natexlab{a}})}\BibitemShut {NoStop}%
\bibitem [{\citenamefont {Moriya}(1960{\natexlab{b}})}]{moriya_prl_1960}%
  \BibitemOpen
  \bibfield  {author} {\bibinfo {author} {\bibfnamefont {T.}~\bibnamefont
  {Moriya}},\ }\bibfield  {title} {\bibinfo {title} {New mechanism of
  anisotropic superexchange interaction},\ }\href
  {https://doi.org/10.1103/PhysRevLett.4.228} {\bibfield  {journal} {\bibinfo
  {journal} {Phys. Rev. Lett.}\ }\textbf {\bibinfo {volume} {4}},\ \bibinfo
  {pages} {228} (\bibinfo {year} {1960}{\natexlab{b}})}\BibitemShut {NoStop}%
\bibitem [{\citenamefont {Dender}\ \emph {et~al.}(1996)\citenamefont {Dender},
  \citenamefont {Davidovi\ifmmode~\acute{c}\else \'{c}\fi{}}, \citenamefont
  {Reich}, \citenamefont {Broholm}, \citenamefont {Lefmann},\ and\
  \citenamefont {Aeppli}}]{comp_1}%
  \BibitemOpen
  \bibfield  {author} {\bibinfo {author} {\bibfnamefont {D.~C.}\ \bibnamefont
  {Dender}}, \bibinfo {author} {\bibfnamefont {D.}~\bibnamefont
  {Davidovi\ifmmode~\acute{c}\else \'{c}\fi{}}}, \bibinfo {author}
  {\bibfnamefont {D.~H.}\ \bibnamefont {Reich}}, \bibinfo {author}
  {\bibfnamefont {C.}~\bibnamefont {Broholm}}, \bibinfo {author} {\bibfnamefont
  {K.}~\bibnamefont {Lefmann}},\ and\ \bibinfo {author} {\bibfnamefont
  {G.}~\bibnamefont {Aeppli}},\ }\bibfield  {title} {\bibinfo {title} {Magnetic
  properties of a quasi-one-dimensional s=1/2 antiferromagnet: Copper
  benzoate},\ }\href {https://doi.org/10.1103/PhysRevB.53.2583} {\bibfield
  {journal} {\bibinfo  {journal} {Phys. Rev. B}\ }\textbf {\bibinfo {volume}
  {53}},\ \bibinfo {pages} {2583} (\bibinfo {year} {1996})}\BibitemShut
  {NoStop}%
\bibitem [{\citenamefont {Dender}\ \emph {et~al.}(1997)\citenamefont {Dender},
  \citenamefont {Hammar}, \citenamefont {Reich}, \citenamefont {Broholm},\ and\
  \citenamefont {Aeppli}}]{comp_2}%
  \BibitemOpen
  \bibfield  {author} {\bibinfo {author} {\bibfnamefont {D.~C.}\ \bibnamefont
  {Dender}}, \bibinfo {author} {\bibfnamefont {P.~R.}\ \bibnamefont {Hammar}},
  \bibinfo {author} {\bibfnamefont {D.~H.}\ \bibnamefont {Reich}}, \bibinfo
  {author} {\bibfnamefont {C.}~\bibnamefont {Broholm}},\ and\ \bibinfo {author}
  {\bibfnamefont {G.}~\bibnamefont {Aeppli}},\ }\bibfield  {title} {\bibinfo
  {title} {Direct observation of field-induced incommensurate fluctuations in a
  one-dimensional
  $\mathit{S}\phantom{\rule{0ex}{0ex}}=\phantom{\rule{0ex}{0ex}}1/2$
  antiferromagnet},\ }\href {https://doi.org/10.1103/PhysRevLett.79.1750}
  {\bibfield  {journal} {\bibinfo  {journal} {Phys. Rev. Lett.}\ }\textbf
  {\bibinfo {volume} {79}},\ \bibinfo {pages} {1750} (\bibinfo {year}
  {1997})}\BibitemShut {NoStop}%
\bibitem [{\citenamefont {Kohgi}\ \emph {et~al.}(2001)\citenamefont {Kohgi},
  \citenamefont {Iwasa}, \citenamefont {Mignot}, \citenamefont {F\aa{}k},
  \citenamefont {Gegenwart}, \citenamefont {Lang}, \citenamefont {Ochiai},
  \citenamefont {Aoki},\ and\ \citenamefont {Suzuki}}]{comp_3}%
  \BibitemOpen
  \bibfield  {author} {\bibinfo {author} {\bibfnamefont {M.}~\bibnamefont
  {Kohgi}}, \bibinfo {author} {\bibfnamefont {K.}~\bibnamefont {Iwasa}},
  \bibinfo {author} {\bibfnamefont {J.-M.}\ \bibnamefont {Mignot}}, \bibinfo
  {author} {\bibfnamefont {B.}~\bibnamefont {F\aa{}k}}, \bibinfo {author}
  {\bibfnamefont {P.}~\bibnamefont {Gegenwart}}, \bibinfo {author}
  {\bibfnamefont {M.}~\bibnamefont {Lang}}, \bibinfo {author} {\bibfnamefont
  {A.}~\bibnamefont {Ochiai}}, \bibinfo {author} {\bibfnamefont
  {H.}~\bibnamefont {Aoki}},\ and\ \bibinfo {author} {\bibfnamefont
  {T.}~\bibnamefont {Suzuki}},\ }\bibfield  {title} {\bibinfo {title}
  {Staggered field effect on the one-dimensional
  $\mathit{S}\phantom{\rule{0ex}{0ex}}=\phantom{\rule{0ex}{0ex}}\frac{1}{2}$
  antiferromagnet ${\mathrm{yb}}_{4}{\mathrm{as}}_{3}$},\ }\href
  {https://doi.org/10.1103/PhysRevLett.86.2439} {\bibfield  {journal} {\bibinfo
   {journal} {Phys. Rev. Lett.}\ }\textbf {\bibinfo {volume} {86}},\ \bibinfo
  {pages} {2439} (\bibinfo {year} {2001})}\BibitemShut {NoStop}%
\bibitem [{\citenamefont {Fulde}\ \emph {et~al.}(1995)\citenamefont {Fulde},
  \citenamefont {Schmidt},\ and\ \citenamefont {Thalmeier}}]{comp_4}%
  \BibitemOpen
  \bibfield  {author} {\bibinfo {author} {\bibfnamefont {P.}~\bibnamefont
  {Fulde}}, \bibinfo {author} {\bibfnamefont {B.}~\bibnamefont {Schmidt}},\
  and\ \bibinfo {author} {\bibfnamefont {P.}~\bibnamefont {Thalmeier}},\
  }\bibfield  {title} {\bibinfo {title} {Theoretical model for the semi-metal
  yb4as3},\ }\href {https://doi.org/10.1209/0295-5075/31/5-6/013} {\bibfield
  {journal} {\bibinfo  {journal} {Europhysics Letters}\ }\textbf {\bibinfo
  {volume} {31}},\ \bibinfo {pages} {323} (\bibinfo {year} {1995})}\BibitemShut
  {NoStop}%
\bibitem [{\citenamefont {Stagraczynski}\ \emph {et~al.}(2024)\citenamefont
  {Stagraczynski}, \citenamefont {Balaz}, \citenamefont {Jafari}, \citenamefont
  {Barnas},\ and\ \citenamefont {Dyrdal}}]{jafari_experimental}%
  \BibitemOpen
  \bibfield  {author} {\bibinfo {author} {\bibfnamefont {S.}~\bibnamefont
  {Stagraczynski}}, \bibinfo {author} {\bibfnamefont {P.}~\bibnamefont
  {Balaz}}, \bibinfo {author} {\bibfnamefont {M.}~\bibnamefont {Jafari}},
  \bibinfo {author} {\bibfnamefont {J.}~\bibnamefont {Barnas}},\ and\ \bibinfo
  {author} {\bibfnamefont {A.}~\bibnamefont {Dyrdal}},\ }\href
  {https://arxiv.org/abs/2404.15543} {\bibinfo {title} {Magnetic ordering and
  dynamics in monolayers and bilayers of chromium trihalides: atomistic
  simulations approach}} (\bibinfo {year} {2024}),\ \Eprint
  {https://arxiv.org/abs/2404.15543} {arXiv:2404.15543 [cond-mat.mes-hall]}
  \BibitemShut {NoStop}%
\bibitem [{\citenamefont {Jafari}\ \emph {et~al.}(2008)\citenamefont {Jafari},
  \citenamefont {Kargarian}, \citenamefont {Langari},\ and\ \citenamefont
  {Siahatgar}}]{jafri_prb_2008}%
  \BibitemOpen
  \bibfield  {author} {\bibinfo {author} {\bibfnamefont {R.}~\bibnamefont
  {Jafari}}, \bibinfo {author} {\bibfnamefont {M.}~\bibnamefont {Kargarian}},
  \bibinfo {author} {\bibfnamefont {A.}~\bibnamefont {Langari}},\ and\ \bibinfo
  {author} {\bibfnamefont {M.}~\bibnamefont {Siahatgar}},\ }\bibfield  {title}
  {\bibinfo {title} {Phase diagram and entanglement of the ising model with
  dzyaloshinskii-moriya interaction},\ }\href
  {https://doi.org/10.1103/PhysRevB.78.214414} {\bibfield  {journal} {\bibinfo
  {journal} {Phys. Rev. B}\ }\textbf {\bibinfo {volume} {78}},\ \bibinfo
  {pages} {214414} (\bibinfo {year} {2008})}\BibitemShut {NoStop}%
\bibitem [{\citenamefont {Luo}(2022{\natexlab{a}})}]{GR_phase_DM}%
  \BibitemOpen
  \bibfield  {author} {\bibinfo {author} {\bibfnamefont {Q.}~\bibnamefont
  {Luo}},\ }\bibfield  {title} {\bibinfo {title} {Analytical results for the
  unusual gr\"uneisen ratio in the quantum ising model with
  dzyaloshinskii-moriya interaction},\ }\href
  {https://doi.org/10.1103/PhysRevB.105.L060401} {\bibfield  {journal}
  {\bibinfo  {journal} {Phys. Rev. B}\ }\textbf {\bibinfo {volume} {105}},\
  \bibinfo {pages} {L060401} (\bibinfo {year}
  {2022}{\natexlab{a}})}\BibitemShut {NoStop}%
\bibitem [{\citenamefont {Kargarian}\ \emph {et~al.}(2009)\citenamefont
  {Kargarian}, \citenamefont {Jafari},\ and\ \citenamefont
  {Langari}}]{jafari_XY_entanglement}%
  \BibitemOpen
  \bibfield  {author} {\bibinfo {author} {\bibfnamefont {M.}~\bibnamefont
  {Kargarian}}, \bibinfo {author} {\bibfnamefont {R.}~\bibnamefont {Jafari}},\
  and\ \bibinfo {author} {\bibfnamefont {A.}~\bibnamefont {Langari}},\
  }\bibfield  {title} {\bibinfo {title} {Dzyaloshinskii-moriya interaction and
  anisotropy effects on the entanglement of the heisenberg model},\ }\href
  {https://doi.org/10.1103/PhysRevA.79.042319} {\bibfield  {journal} {\bibinfo
  {journal} {Phys. Rev. A}\ }\textbf {\bibinfo {volume} {79}},\ \bibinfo
  {pages} {042319} (\bibinfo {year} {2009})}\BibitemShut {NoStop}%
\bibitem [{\citenamefont {Roy}\ \emph {et~al.}(2019)\citenamefont {Roy},
  \citenamefont {Chanda}, \citenamefont {Das}, \citenamefont {Sadhukhan},
  \citenamefont {Sen(De)},\ and\ \citenamefont {Sen}}]{roy_prb_2019}%
  \BibitemOpen
  \bibfield  {author} {\bibinfo {author} {\bibfnamefont {S.}~\bibnamefont
  {Roy}}, \bibinfo {author} {\bibfnamefont {T.}~\bibnamefont {Chanda}},
  \bibinfo {author} {\bibfnamefont {T.}~\bibnamefont {Das}}, \bibinfo {author}
  {\bibfnamefont {D.}~\bibnamefont {Sadhukhan}}, \bibinfo {author}
  {\bibfnamefont {A.}~\bibnamefont {Sen(De)}},\ and\ \bibinfo {author}
  {\bibfnamefont {U.}~\bibnamefont {Sen}},\ }\bibfield  {title} {\bibinfo
  {title} {Phase boundaries in an alternating-field quantum xy model with
  dzyaloshinskii-moriya interaction: Sustainable entanglement in dynamics},\
  }\href {https://doi.org/10.1103/PhysRevB.99.064422} {\bibfield  {journal}
  {\bibinfo  {journal} {Phys. Rev. B}\ }\textbf {\bibinfo {volume} {99}},\
  \bibinfo {pages} {064422} (\bibinfo {year} {2019})}\BibitemShut {NoStop}%
\bibitem [{\citenamefont {Ming}\ \emph {et~al.}(2013)\citenamefont {Ming},
  \citenamefont {Hui}, \citenamefont {Xiao-Xian},\ and\ \citenamefont
  {Pei-Qing}}]{Zhong_2013}%
  \BibitemOpen
  \bibfield  {author} {\bibinfo {author} {\bibfnamefont {Z.}~\bibnamefont
  {Ming}}, \bibinfo {author} {\bibfnamefont {X.}~\bibnamefont {Hui}}, \bibinfo
  {author} {\bibfnamefont {L.}~\bibnamefont {Xiao-Xian}},\ and\ \bibinfo
  {author} {\bibfnamefont {T.}~\bibnamefont {Pei-Qing}},\ }\bibfield  {title}
  {\bibinfo {title} {{The effects of the
  Dzyaloshinskii{\ifmmode---\else\textemdash\fi}Moriya interaction on the
  ground-state properties of the XY chain in a transverse field}},\ }\href
  {https://doi.org/10.1088/1674-1056/22/9/090313} {\bibfield  {journal}
  {\bibinfo  {journal} {Chin. Phys. B}\ }\textbf {\bibinfo {volume} {22}},\
  \bibinfo {pages} {090313} (\bibinfo {year} {2013})}\BibitemShut {NoStop}%
\bibitem [{\citenamefont {Liu}\ \emph {et~al.}(2020)\citenamefont {Liu},
  \citenamefont {Yi}, \citenamefont {Sun}, \citenamefont {Dong},\ and\
  \citenamefont {You}}]{quantum_phase_DM_2}%
  \BibitemOpen
  \bibfield  {author} {\bibinfo {author} {\bibfnamefont {Z.-A.}\ \bibnamefont
  {Liu}}, \bibinfo {author} {\bibfnamefont {T.-C.}\ \bibnamefont {Yi}},
  \bibinfo {author} {\bibfnamefont {J.-H.}\ \bibnamefont {Sun}}, \bibinfo
  {author} {\bibfnamefont {Y.-L.}\ \bibnamefont {Dong}},\ and\ \bibinfo
  {author} {\bibfnamefont {W.-L.}\ \bibnamefont {You}},\ }\bibfield  {title}
  {\bibinfo {title} {Lifshitz phase transitions in a one-dimensional gamma
  model},\ }\href {https://doi.org/10.1103/PhysRevE.102.032127} {\bibfield
  {journal} {\bibinfo  {journal} {Phys. Rev. E}\ }\textbf {\bibinfo {volume}
  {102}},\ \bibinfo {pages} {032127} (\bibinfo {year} {2020})}\BibitemShut
  {NoStop}%
\bibitem [{\citenamefont {Hou{\ifmmode\mbox{\c{c}}\else\c{c}\fi}a}\ \emph
  {et~al.}(2022)\citenamefont {Hou{\ifmmode\mbox{\c{c}}\else\c{c}\fi}a},
  \citenamefont {Belouad}, \citenamefont {Choubabi}, \citenamefont {Kamal},\
  and\ \citenamefont {El~Bouziani}}]{DM_teleportation}%
  \BibitemOpen
  \bibfield  {author} {\bibinfo {author} {\bibfnamefont {R.}~\bibnamefont
  {Hou{\ifmmode\mbox{\c{c}}\else\c{c}\fi}a}}, \bibinfo {author} {\bibfnamefont
  {A.}~\bibnamefont {Belouad}}, \bibinfo {author} {\bibfnamefont {E.~B.}\
  \bibnamefont {Choubabi}}, \bibinfo {author} {\bibfnamefont {A.}~\bibnamefont
  {Kamal}},\ and\ \bibinfo {author} {\bibfnamefont {M.}~\bibnamefont
  {El~Bouziani}},\ }\bibfield  {title} {\bibinfo {title} {{Quantum
  teleportation via a two-qubit Heisenberg XXX chain with x-component of
  Dzyaloshinskii{\textendash}Moriya interaction}},\ }\href
  {https://doi.org/10.1016/j.jmmm.2022.169816} {\bibfield  {journal} {\bibinfo
  {journal} {J. Magn. Magn. Mater.}\ }\textbf {\bibinfo {volume} {563}},\
  \bibinfo {pages} {169816} (\bibinfo {year} {2022})}\BibitemShut {NoStop}%
\bibitem [{\citenamefont {Wang}\ \emph {et~al.}(2012)\citenamefont {Wang},
  \citenamefont {Wu},\ and\ \citenamefont {Chen}}]{thermal_entanglement}%
  \BibitemOpen
  \bibfield  {author} {\bibinfo {author} {\bibfnamefont {H.}~\bibnamefont
  {Wang}}, \bibinfo {author} {\bibfnamefont {G.}~\bibnamefont {Wu}},\ and\
  \bibinfo {author} {\bibfnamefont {D.}~\bibnamefont {Chen}},\ }\bibfield
  {title} {\bibinfo {title} {Thermal entangled quantum otto engine based on the
  two qubits heisenberg model with dzyaloshinskii–moriya interaction in an
  external magnetic field},\ }\href
  {https://doi.org/10.1088/0031-8949/86/01/015001} {\bibfield  {journal}
  {\bibinfo  {journal} {Physica Scripta}\ }\textbf {\bibinfo {volume} {86}},\
  \bibinfo {pages} {015001} (\bibinfo {year} {2012})}\BibitemShut {NoStop}%
\bibitem [{\citenamefont {JAFARI}\ and\ \citenamefont
  {LANGARI}(2011)}]{jafari_thermal_gamma}%
  \BibitemOpen
  \bibfield  {author} {\bibinfo {author} {\bibfnamefont {R.}~\bibnamefont
  {JAFARI}}\ and\ \bibinfo {author} {\bibfnamefont {A.}~\bibnamefont
  {LANGARI}},\ }\bibfield  {title} {\bibinfo {title} {Three-qubits ground state
  and thermal entanglement of anisotropic heisenberg (xxz) and ising models
  with dzyaloshinskii-moriya interaction},\ }\href
  {https://doi.org/10.1142/S0219749911007800} {\bibfield  {journal} {\bibinfo
  {journal} {International Journal of Quantum Information}\ }\textbf {\bibinfo
  {volume} {09}},\ \bibinfo {pages} {1057} (\bibinfo {year} {2011})},\ \Eprint
  {https://arxiv.org/abs/https://doi.org/10.1142/S0219749911007800}
  {https://doi.org/10.1142/S0219749911007800} \BibitemShut {NoStop}%
\bibitem [{\citenamefont {Mehran}\ \emph {et~al.}(2014)\citenamefont {Mehran},
  \citenamefont {Mahdavifar},\ and\ \citenamefont
  {Jafari}}]{jafari_thermal_heisenberg}%
  \BibitemOpen
  \bibfield  {author} {\bibinfo {author} {\bibfnamefont {E.}~\bibnamefont
  {Mehran}}, \bibinfo {author} {\bibfnamefont {S.}~\bibnamefont {Mahdavifar}},\
  and\ \bibinfo {author} {\bibfnamefont {R.}~\bibnamefont {Jafari}},\
  }\bibfield  {title} {\bibinfo {title} {Induced effects of the
  dzyaloshinskii-moriya interaction on the thermal entanglement in spin-1/2
  heisenberg chains},\ }\href {https://doi.org/10.1103/PhysRevA.89.042306}
  {\bibfield  {journal} {\bibinfo  {journal} {Phys. Rev. A}\ }\textbf {\bibinfo
  {volume} {89}},\ \bibinfo {pages} {042306} (\bibinfo {year}
  {2014})}\BibitemShut {NoStop}%
\bibitem [{\citenamefont {Azimi}\ \emph {et~al.}(2014)\citenamefont {Azimi},
  \citenamefont {Chotorlishvili}, \citenamefont {Mishra}, \citenamefont
  {Vekua}, \citenamefont {Hübner},\ and\ \citenamefont
  {Berakdar}}]{heat_engine_DM_1}%
  \BibitemOpen
  \bibfield  {author} {\bibinfo {author} {\bibfnamefont {M.}~\bibnamefont
  {Azimi}}, \bibinfo {author} {\bibfnamefont {L.}~\bibnamefont
  {Chotorlishvili}}, \bibinfo {author} {\bibfnamefont {S.~K.}\ \bibnamefont
  {Mishra}}, \bibinfo {author} {\bibfnamefont {T.}~\bibnamefont {Vekua}},
  \bibinfo {author} {\bibfnamefont {W.}~\bibnamefont {Hübner}},\ and\ \bibinfo
  {author} {\bibfnamefont {J.}~\bibnamefont {Berakdar}},\ }\bibfield  {title}
  {\bibinfo {title} {Quantum otto heat engine based on a multiferroic chain
  working substance},\ }\href {https://doi.org/10.1088/1367-2630/16/6/063018}
  {\bibfield  {journal} {\bibinfo  {journal} {New Journal of Physics}\ }\textbf
  {\bibinfo {volume} {16}},\ \bibinfo {pages} {063018} (\bibinfo {year}
  {2014})}\BibitemShut {NoStop}%
\bibitem [{\citenamefont {Chotorlishvili}\ \emph {et~al.}(2016)\citenamefont
  {Chotorlishvili}, \citenamefont {Azimi}, \citenamefont
  {Stagraczy\ifmmode~\acute{n}\else \'{n}\fi{}ski}, \citenamefont
  {Toklikishvili}, \citenamefont {Sch\"uler},\ and\ \citenamefont
  {Berakdar}}]{heat_engine_DM_2}%
  \BibitemOpen
  \bibfield  {author} {\bibinfo {author} {\bibfnamefont {L.}~\bibnamefont
  {Chotorlishvili}}, \bibinfo {author} {\bibfnamefont {M.}~\bibnamefont
  {Azimi}}, \bibinfo {author} {\bibfnamefont {S.}~\bibnamefont
  {Stagraczy\ifmmode~\acute{n}\else \'{n}\fi{}ski}}, \bibinfo {author}
  {\bibfnamefont {Z.}~\bibnamefont {Toklikishvili}}, \bibinfo {author}
  {\bibfnamefont {M.}~\bibnamefont {Sch\"uler}},\ and\ \bibinfo {author}
  {\bibfnamefont {J.}~\bibnamefont {Berakdar}},\ }\bibfield  {title} {\bibinfo
  {title} {Superadiabatic quantum heat engine with a multiferroic working
  medium},\ }\href {https://doi.org/10.1103/PhysRevE.94.032116} {\bibfield
  {journal} {\bibinfo  {journal} {Phys. Rev. E}\ }\textbf {\bibinfo {volume}
  {94}},\ \bibinfo {pages} {032116} (\bibinfo {year} {2016})}\BibitemShut
  {NoStop}%
\bibitem [{\citenamefont {Wang}\ \emph {et~al.}(2018)\citenamefont {Wang},
  \citenamefont {Cao},\ and\ \citenamefont {Quan}}]{nonequi_thermo_DM}%
  \BibitemOpen
  \bibfield  {author} {\bibinfo {author} {\bibfnamefont {Q.}~\bibnamefont
  {Wang}}, \bibinfo {author} {\bibfnamefont {D.}~\bibnamefont {Cao}},\ and\
  \bibinfo {author} {\bibfnamefont {H.~T.}\ \bibnamefont {Quan}},\ }\bibfield
  {title} {\bibinfo {title} {Effects of the dzyaloshinsky-moriya interaction on
  nonequilibrium thermodynamics in the $xy$ chain in a transverse field},\
  }\href {https://doi.org/10.1103/PhysRevE.98.022107} {\bibfield  {journal}
  {\bibinfo  {journal} {Phys. Rev. E}\ }\textbf {\bibinfo {volume} {98}},\
  \bibinfo {pages} {022107} (\bibinfo {year} {2018})}\BibitemShut {NoStop}%
\bibitem [{\citenamefont {Azimi}\ \emph {et~al.}(2016)\citenamefont {Azimi},
  \citenamefont {Sekania}, \citenamefont {Mishra}, \citenamefont
  {Chotorlishvili}, \citenamefont {Toklikishvili},\ and\ \citenamefont
  {Berakdar}}]{nonequi_DM}%
  \BibitemOpen
  \bibfield  {author} {\bibinfo {author} {\bibfnamefont {M.}~\bibnamefont
  {Azimi}}, \bibinfo {author} {\bibfnamefont {M.}~\bibnamefont {Sekania}},
  \bibinfo {author} {\bibfnamefont {S.~K.}\ \bibnamefont {Mishra}}, \bibinfo
  {author} {\bibfnamefont {L.}~\bibnamefont {Chotorlishvili}}, \bibinfo
  {author} {\bibfnamefont {Z.}~\bibnamefont {Toklikishvili}},\ and\ \bibinfo
  {author} {\bibfnamefont {J.}~\bibnamefont {Berakdar}},\ }\bibfield  {title}
  {\bibinfo {title} {Pulse and quench induced dynamical phase transition in a
  chiral multiferroic spin chain},\ }\href
  {https://doi.org/10.1103/PhysRevB.94.064423} {\bibfield  {journal} {\bibinfo
  {journal} {Phys. Rev. B}\ }\textbf {\bibinfo {volume} {94}},\ \bibinfo
  {pages} {064423} (\bibinfo {year} {2016})}\BibitemShut {NoStop}%
\bibitem [{\citenamefont {Farajollahpour}\ and\ \citenamefont
  {Jafari}(2018)}]{jafari_topolpgy_DM}%
  \BibitemOpen
  \bibfield  {author} {\bibinfo {author} {\bibfnamefont {T.}~\bibnamefont
  {Farajollahpour}}\ and\ \bibinfo {author} {\bibfnamefont {S.~A.}\
  \bibnamefont {Jafari}},\ }\bibfield  {title} {\bibinfo {title} {Topological
  phase transition of the anisotropic $xy$ model with dzyaloshinskii-moriya
  interaction},\ }\href {https://doi.org/10.1103/PhysRevB.98.085136} {\bibfield
   {journal} {\bibinfo  {journal} {Phys. Rev. B}\ }\textbf {\bibinfo {volume}
  {98}},\ \bibinfo {pages} {085136} (\bibinfo {year} {2018})}\BibitemShut
  {NoStop}%
\bibitem [{\citenamefont {Zhu}\ \emph {et~al.}(2023)\citenamefont {Zhu},
  \citenamefont {Wang}, \citenamefont {Shao}, \citenamefont {Zou},\ and\
  \citenamefont {Wu}}]{qsl_DM}%
  \BibitemOpen
  \bibfield  {author} {\bibinfo {author} {\bibfnamefont {Z.-R.}\ \bibnamefont
  {Zhu}}, \bibinfo {author} {\bibfnamefont {Q.}~\bibnamefont {Wang}}, \bibinfo
  {author} {\bibfnamefont {B.}~\bibnamefont {Shao}}, \bibinfo {author}
  {\bibfnamefont {J.}~\bibnamefont {Zou}},\ and\ \bibinfo {author}
  {\bibfnamefont {L.-A.}\ \bibnamefont {Wu}},\ }\bibfield  {title} {\bibinfo
  {title} {Effect of the dzyaloshinskii-moriya interaction on quantum speed
  limit and orthogonality catastrophe},\ }\href
  {https://doi.org/10.1103/PhysRevA.107.042427} {\bibfield  {journal} {\bibinfo
   {journal} {Phys. Rev. A}\ }\textbf {\bibinfo {volume} {107}},\ \bibinfo
  {pages} {042427} (\bibinfo {year} {2023})}\BibitemShut {NoStop}%
\bibitem [{\citenamefont {Sen}\ \emph {et~al.}(2016)\citenamefont {Sen},
  \citenamefont {Nandy},\ and\ \citenamefont
  {Sengupta}}]{dyn_relax_floquet_2016}%
  \BibitemOpen
  \bibfield  {author} {\bibinfo {author} {\bibfnamefont {A.}~\bibnamefont
  {Sen}}, \bibinfo {author} {\bibfnamefont {S.}~\bibnamefont {Nandy}},\ and\
  \bibinfo {author} {\bibfnamefont {K.}~\bibnamefont {Sengupta}},\ }\bibfield
  {title} {\bibinfo {title} {Entanglement generation in periodically driven
  integrable systems: Dynamical phase transitions and steady state},\ }\href
  {https://doi.org/10.1103/PhysRevB.94.214301} {\bibfield  {journal} {\bibinfo
  {journal} {Phys. Rev. B}\ }\textbf {\bibinfo {volume} {94}},\ \bibinfo
  {pages} {214301} (\bibinfo {year} {2016})}\BibitemShut {NoStop}%
\bibitem [{\citenamefont {Nandy}\ \emph {et~al.}(2018)\citenamefont {Nandy},
  \citenamefont {Sengupta},\ and\ \citenamefont {Sen}}]{Nandy_2018}%
  \BibitemOpen
  \bibfield  {author} {\bibinfo {author} {\bibfnamefont {S.}~\bibnamefont
  {Nandy}}, \bibinfo {author} {\bibfnamefont {K.}~\bibnamefont {Sengupta}},\
  and\ \bibinfo {author} {\bibfnamefont {A.}~\bibnamefont {Sen}},\ }\bibfield
  {title} {\bibinfo {title} {Periodically driven integrable systems with
  long-range pair potentials},\ }\href
  {https://doi.org/10.1088/1751-8121/aaced6} {\bibfield  {journal} {\bibinfo
  {journal} {Journal of Physics A: Mathematical and Theoretical}\ }\textbf
  {\bibinfo {volume} {51}},\ \bibinfo {pages} {334002} (\bibinfo {year}
  {2018})}\BibitemShut {NoStop}%
\bibitem [{\citenamefont {Sarkar}\ and\ \citenamefont
  {Sengupta}(2020)}]{dyn_relax_2020}%
  \BibitemOpen
  \bibfield  {author} {\bibinfo {author} {\bibfnamefont {M.}~\bibnamefont
  {Sarkar}}\ and\ \bibinfo {author} {\bibfnamefont {K.}~\bibnamefont
  {Sengupta}},\ }\bibfield  {title} {\bibinfo {title} {Dynamical transition for
  a class of integrable models coupled to a bath},\ }\href
  {https://doi.org/10.1103/PhysRevB.102.235154} {\bibfield  {journal} {\bibinfo
   {journal} {Phys. Rev. B}\ }\textbf {\bibinfo {volume} {102}},\ \bibinfo
  {pages} {235154} (\bibinfo {year} {2020})}\BibitemShut {NoStop}%
\bibitem [{\citenamefont {Aditya}\ \emph {et~al.}(2022)\citenamefont {Aditya},
  \citenamefont {Samanta}, \citenamefont {Sen}, \citenamefont {Sengupta},\ and\
  \citenamefont {Sen}}]{dyn_relax_2022}%
  \BibitemOpen
  \bibfield  {author} {\bibinfo {author} {\bibfnamefont {S.}~\bibnamefont
  {Aditya}}, \bibinfo {author} {\bibfnamefont {S.}~\bibnamefont {Samanta}},
  \bibinfo {author} {\bibfnamefont {A.}~\bibnamefont {Sen}}, \bibinfo {author}
  {\bibfnamefont {K.}~\bibnamefont {Sengupta}},\ and\ \bibinfo {author}
  {\bibfnamefont {D.}~\bibnamefont {Sen}},\ }\bibfield  {title} {\bibinfo
  {title} {Dynamical relaxation of correlators in periodically driven
  integrable quantum systems},\ }\href
  {https://doi.org/10.1103/PhysRevB.105.104303} {\bibfield  {journal} {\bibinfo
   {journal} {Phys. Rev. B}\ }\textbf {\bibinfo {volume} {105}},\ \bibinfo
  {pages} {104303} (\bibinfo {year} {2022})}\BibitemShut {NoStop}%
\bibitem [{\citenamefont {Bernien}\ \emph {et~al.}(2017)\citenamefont
  {Bernien}, \citenamefont {Schwartz}, \citenamefont {Keesling}, \citenamefont
  {Levine}, \citenamefont {Omran}, \citenamefont {Pichler}, \citenamefont
  {Choi}, \citenamefont {Zibrov}, \citenamefont {Endres}, \citenamefont
  {Greiner}, \citenamefont {Vuleti{\ifmmode\acute{c}\else\'{c}\fi}},\ and\
  \citenamefont {Lukin}}]{oscillation_dyn_phases_2017}%
  \BibitemOpen
  \bibfield  {author} {\bibinfo {author} {\bibfnamefont {H.}~\bibnamefont
  {Bernien}}, \bibinfo {author} {\bibfnamefont {S.}~\bibnamefont {Schwartz}},
  \bibinfo {author} {\bibfnamefont {A.}~\bibnamefont {Keesling}}, \bibinfo
  {author} {\bibfnamefont {H.}~\bibnamefont {Levine}}, \bibinfo {author}
  {\bibfnamefont {A.}~\bibnamefont {Omran}}, \bibinfo {author} {\bibfnamefont
  {H.}~\bibnamefont {Pichler}}, \bibinfo {author} {\bibfnamefont
  {S.}~\bibnamefont {Choi}}, \bibinfo {author} {\bibfnamefont {A.~S.}\
  \bibnamefont {Zibrov}}, \bibinfo {author} {\bibfnamefont {M.}~\bibnamefont
  {Endres}}, \bibinfo {author} {\bibfnamefont {M.}~\bibnamefont {Greiner}},
  \bibinfo {author} {\bibfnamefont {V.}~\bibnamefont
  {Vuleti{\ifmmode\acute{c}\else\'{c}\fi}}},\ and\ \bibinfo {author}
  {\bibfnamefont {M.~D.}\ \bibnamefont {Lukin}},\ }\bibfield  {title} {\bibinfo
  {title} {{Probing many-body dynamics on a 51-atom quantum simulator}},\
  }\href {https://doi.org/10.1038/nature24622} {\bibfield  {journal} {\bibinfo
  {journal} {Nature}\ }\textbf {\bibinfo {volume} {551}},\ \bibinfo {pages}
  {579} (\bibinfo {year} {2017})}\BibitemShut {NoStop}%
\bibitem [{\citenamefont {Delfino}\ and\ \citenamefont
  {Sorba}(2022)}]{DELFINO2022115643}%
  \BibitemOpen
  \bibfield  {author} {\bibinfo {author} {\bibfnamefont {G.}~\bibnamefont
  {Delfino}}\ and\ \bibinfo {author} {\bibfnamefont {M.}~\bibnamefont
  {Sorba}},\ }\bibfield  {title} {\bibinfo {title} {Persistent oscillations
  after quantum quenches in d dimensions},\ }\href
  {https://doi.org/https://doi.org/10.1016/j.nuclphysb.2021.115643} {\bibfield
  {journal} {\bibinfo  {journal} {Nuclear Physics B}\ }\textbf {\bibinfo
  {volume} {974}},\ \bibinfo {pages} {115643} (\bibinfo {year}
  {2022})}\BibitemShut {NoStop}%
\bibitem [{\citenamefont {Castro-Alvaredo}\ \emph {et~al.}(2020)\citenamefont
  {Castro-Alvaredo}, \citenamefont {Lencs\'es}, \citenamefont {Sz\'ecs\'enyi},\
  and\ \citenamefont {Viti}}]{ent_oscill_2020}%
  \BibitemOpen
  \bibfield  {author} {\bibinfo {author} {\bibfnamefont {O.~A.}\ \bibnamefont
  {Castro-Alvaredo}}, \bibinfo {author} {\bibfnamefont {M.}~\bibnamefont
  {Lencs\'es}}, \bibinfo {author} {\bibfnamefont {I.~M.}\ \bibnamefont
  {Sz\'ecs\'enyi}},\ and\ \bibinfo {author} {\bibfnamefont {J.}~\bibnamefont
  {Viti}},\ }\bibfield  {title} {\bibinfo {title} {Entanglement oscillations
  near a quantum critical point},\ }\href
  {https://doi.org/10.1103/PhysRevLett.124.230601} {\bibfield  {journal}
  {\bibinfo  {journal} {Phys. Rev. Lett.}\ }\textbf {\bibinfo {volume} {124}},\
  \bibinfo {pages} {230601} (\bibinfo {year} {2020})}\BibitemShut {NoStop}%
\bibitem [{\citenamefont {Heyl}\ \emph {et~al.}(2013)\citenamefont {Heyl},
  \citenamefont {Polkovnikov},\ and\ \citenamefont {Kehrein}}]{heyl_2013}%
  \BibitemOpen
  \bibfield  {author} {\bibinfo {author} {\bibfnamefont {M.}~\bibnamefont
  {Heyl}}, \bibinfo {author} {\bibfnamefont {A.}~\bibnamefont {Polkovnikov}},\
  and\ \bibinfo {author} {\bibfnamefont {S.}~\bibnamefont {Kehrein}},\
  }\bibfield  {title} {\bibinfo {title} {Dynamical quantum phase transitions in
  the transverse-field ising model},\ }\href
  {https://doi.org/10.1103/PhysRevLett.110.135704} {\bibfield  {journal}
  {\bibinfo  {journal} {Phys. Rev. Lett.}\ }\textbf {\bibinfo {volume} {110}},\
  \bibinfo {pages} {135704} (\bibinfo {year} {2013})}\BibitemShut {NoStop}%
\bibitem [{\citenamefont {Heyl}(2018)}]{Heyl_2018}%
  \BibitemOpen
  \bibfield  {author} {\bibinfo {author} {\bibfnamefont {M.}~\bibnamefont
  {Heyl}},\ }\bibfield  {title} {\bibinfo {title} {Dynamical quantum phase
  transitions: a review},\ }\href {https://doi.org/10.1088/1361-6633/aaaf9a}
  {\bibfield  {journal} {\bibinfo  {journal} {Reports on Progress in Physics}\
  }\textbf {\bibinfo {volume} {81}},\ \bibinfo {pages} {054001} (\bibinfo
  {year} {2018})}\BibitemShut {NoStop}%
\bibitem [{\citenamefont {Makki}\ \emph {et~al.}(2022)\citenamefont {Makki},
  \citenamefont {Bandyopadhyay}, \citenamefont {Maity},\ and\ \citenamefont
  {Dutta}}]{makki_prb_2022}%
  \BibitemOpen
  \bibfield  {author} {\bibinfo {author} {\bibfnamefont {A.~A.}\ \bibnamefont
  {Makki}}, \bibinfo {author} {\bibfnamefont {S.}~\bibnamefont
  {Bandyopadhyay}}, \bibinfo {author} {\bibfnamefont {S.}~\bibnamefont
  {Maity}},\ and\ \bibinfo {author} {\bibfnamefont {A.}~\bibnamefont {Dutta}},\
  }\bibfield  {title} {\bibinfo {title} {Dynamical crossover behavior in the
  relaxation of quenched quantum many-body systems},\ }\href
  {https://doi.org/10.1103/PhysRevB.105.054301} {\bibfield  {journal} {\bibinfo
   {journal} {Phys. Rev. B}\ }\textbf {\bibinfo {volume} {105}},\ \bibinfo
  {pages} {054301} (\bibinfo {year} {2022})}\BibitemShut {NoStop}%
\bibitem [{\citenamefont {Ramos}\ \emph {et~al.}(2023)\citenamefont {Ramos},
  \citenamefont {Urichuk}, \citenamefont {Schneider},\ and\ \citenamefont
  {Sirker}}]{Sirker_XXZ_2023}%
  \BibitemOpen
  \bibfield  {author} {\bibinfo {author} {\bibfnamefont {F.~B.}\ \bibnamefont
  {Ramos}}, \bibinfo {author} {\bibfnamefont {A.}~\bibnamefont {Urichuk}},
  \bibinfo {author} {\bibfnamefont {I.}~\bibnamefont {Schneider}},\ and\
  \bibinfo {author} {\bibfnamefont {J.}~\bibnamefont {Sirker}},\ }\bibfield
  {title} {\bibinfo {title} {Power-law decay of correlations after a global
  quench in the massive xxz chain},\ }\href
  {https://doi.org/10.1103/PhysRevB.107.075138} {\bibfield  {journal} {\bibinfo
   {journal} {Phys. Rev. B}\ }\textbf {\bibinfo {volume} {107}},\ \bibinfo
  {pages} {075138} (\bibinfo {year} {2023})}\BibitemShut {NoStop}%
\bibitem [{\citenamefont {Lakkaraju}\ \emph {et~al.}(2023)\citenamefont
  {Lakkaraju}, \citenamefont {Ghosh}, \citenamefont {Sadhukhan},\ and\
  \citenamefont {De}}]{lakkaraju2023frameworkdynamicaltransitionslongrange}%
  \BibitemOpen
  \bibfield  {author} {\bibinfo {author} {\bibfnamefont {L.~G.~C.}\
  \bibnamefont {Lakkaraju}}, \bibinfo {author} {\bibfnamefont {S.}~\bibnamefont
  {Ghosh}}, \bibinfo {author} {\bibfnamefont {D.}~\bibnamefont {Sadhukhan}},\
  and\ \bibinfo {author} {\bibfnamefont {A.~S.}\ \bibnamefont {De}},\ }\href
  {https://arxiv.org/abs/2305.02945} {\bibinfo {title} {Framework of dynamical
  transitions from long-range to short-range quantum systems}} (\bibinfo {year}
  {2023}),\ \Eprint {https://arxiv.org/abs/2305.02945} {arXiv:2305.02945
  [quant-ph]} \BibitemShut {NoStop}%
\bibitem [{\citenamefont {Sadhukhan}\ \emph {et~al.}(2020)\citenamefont
  {Sadhukhan}, \citenamefont {Sinha}, \citenamefont {Francuz}, \citenamefont
  {Stefaniak}, \citenamefont {Rams}, \citenamefont {Dziarmaga},\ and\
  \citenamefont {Zurek}}]{sadhukhan_prb_2020}%
  \BibitemOpen
  \bibfield  {author} {\bibinfo {author} {\bibfnamefont {D.}~\bibnamefont
  {Sadhukhan}}, \bibinfo {author} {\bibfnamefont {A.}~\bibnamefont {Sinha}},
  \bibinfo {author} {\bibfnamefont {A.}~\bibnamefont {Francuz}}, \bibinfo
  {author} {\bibfnamefont {J.}~\bibnamefont {Stefaniak}}, \bibinfo {author}
  {\bibfnamefont {M.~M.}\ \bibnamefont {Rams}}, \bibinfo {author}
  {\bibfnamefont {J.}~\bibnamefont {Dziarmaga}},\ and\ \bibinfo {author}
  {\bibfnamefont {W.~H.}\ \bibnamefont {Zurek}},\ }\bibfield  {title} {\bibinfo
  {title} {Sonic horizons and causality in phase transition dynamics},\ }\href
  {https://doi.org/10.1103/PhysRevB.101.144429} {\bibfield  {journal} {\bibinfo
   {journal} {Phys. Rev. B}\ }\textbf {\bibinfo {volume} {101}},\ \bibinfo
  {pages} {144429} (\bibinfo {year} {2020})}\BibitemShut {NoStop}%
\bibitem [{\citenamefont {Sinha}\ \emph {et~al.}(2020)\citenamefont {Sinha},
  \citenamefont {Sadhukhan}, \citenamefont {Rams},\ and\ \citenamefont
  {Dziarmaga}}]{sinha_prb_2020}%
  \BibitemOpen
  \bibfield  {author} {\bibinfo {author} {\bibfnamefont {A.}~\bibnamefont
  {Sinha}}, \bibinfo {author} {\bibfnamefont {D.}~\bibnamefont {Sadhukhan}},
  \bibinfo {author} {\bibfnamefont {M.~M.}\ \bibnamefont {Rams}},\ and\
  \bibinfo {author} {\bibfnamefont {J.}~\bibnamefont {Dziarmaga}},\ }\bibfield
  {title} {\bibinfo {title} {Inhomogeneity induced shortcut to adiabaticity in
  ising chains with long-range interactions},\ }\href
  {https://doi.org/10.1103/PhysRevB.102.214203} {\bibfield  {journal} {\bibinfo
   {journal} {Phys. Rev. B}\ }\textbf {\bibinfo {volume} {102}},\ \bibinfo
  {pages} {214203} (\bibinfo {year} {2020})}\BibitemShut {NoStop}%
\bibitem [{\citenamefont {Kac}\ \emph {et~al.}(1963)\citenamefont {Kac},
  \citenamefont {Uhlenbeck},\ and\ \citenamefont {Hemmer}}]{Kac_jmp_1963}%
  \BibitemOpen
  \bibfield  {author} {\bibinfo {author} {\bibfnamefont {M.}~\bibnamefont
  {Kac}}, \bibinfo {author} {\bibfnamefont {G.~E.}\ \bibnamefont {Uhlenbeck}},\
  and\ \bibinfo {author} {\bibfnamefont {P.~C.}\ \bibnamefont {Hemmer}},\
  }\bibfield  {title} {\bibinfo {title} {{On the van der Waals Theory of the
  Vapor{-}Liquid Equilibrium. I. Discussion of a One{-}Dimensional Model}},\
  }\href {https://doi.org/10.1063/1.1703946} {\bibfield  {journal} {\bibinfo
  {journal} {J. Math. Phys.}\ }\textbf {\bibinfo {volume} {4}},\ \bibinfo
  {pages} {216} (\bibinfo {year} {1963})}\BibitemShut {NoStop}%
\bibitem [{\citenamefont {Barouch}\ \emph {et~al.}(1970)\citenamefont
  {Barouch}, \citenamefont {McCoy},\ and\ \citenamefont
  {Dresden}}]{barouch_pra_1970_1}%
  \BibitemOpen
  \bibfield  {author} {\bibinfo {author} {\bibfnamefont {E.}~\bibnamefont
  {Barouch}}, \bibinfo {author} {\bibfnamefont {B.~M.}\ \bibnamefont {McCoy}},\
  and\ \bibinfo {author} {\bibfnamefont {M.}~\bibnamefont {Dresden}},\
  }\bibfield  {title} {\bibinfo {title} {Statistical mechanics of the
  $\mathrm{XY}$ model. i},\ }\href {https://doi.org/10.1103/PhysRevA.2.1075}
  {\bibfield  {journal} {\bibinfo  {journal} {Phys. Rev. A}\ }\textbf {\bibinfo
  {volume} {2}},\ \bibinfo {pages} {1075} (\bibinfo {year} {1970})}\BibitemShut
  {NoStop}%
\bibitem [{\citenamefont {Barouch}\ and\ \citenamefont
  {McCoy}(1971)}]{barouch_pra_1970_2}%
  \BibitemOpen
  \bibfield  {author} {\bibinfo {author} {\bibfnamefont {E.}~\bibnamefont
  {Barouch}}\ and\ \bibinfo {author} {\bibfnamefont {B.~M.}\ \bibnamefont
  {McCoy}},\ }\bibfield  {title} {\bibinfo {title} {Statistical mechanics of
  the $xy$ model. ii. spin-correlation functions},\ }\href
  {https://doi.org/10.1103/PhysRevA.3.786} {\bibfield  {journal} {\bibinfo
  {journal} {Phys. Rev. A}\ }\textbf {\bibinfo {volume} {3}},\ \bibinfo {pages}
  {786} (\bibinfo {year} {1971})}\BibitemShut {NoStop}%
\bibitem [{\citenamefont {Lieb}\ \emph {et~al.}(1961)\citenamefont {Lieb},
  \citenamefont {Schultz},\ and\ \citenamefont {Mattis}}]{lieb1961}%
  \BibitemOpen
  \bibfield  {author} {\bibinfo {author} {\bibfnamefont {E.}~\bibnamefont
  {Lieb}}, \bibinfo {author} {\bibfnamefont {T.}~\bibnamefont {Schultz}},\ and\
  \bibinfo {author} {\bibfnamefont {D.}~\bibnamefont {Mattis}},\ }\bibfield
  {title} {\bibinfo {title} {Two soluble models of an antiferromagnetic
  chain},\ }\href
  {https://doi.org/https://doi.org/10.1016/0003-4916(61)90115-4} {\bibfield
  {journal} {\bibinfo  {journal} {Annals of Physics}\ }\textbf {\bibinfo
  {volume} {16}},\ \bibinfo {pages} {407} (\bibinfo {year} {1961})}\BibitemShut
  {NoStop}%
\bibitem [{\citenamefont {Mbeng}\ \emph {et~al.}(2024)\citenamefont {Mbeng},
  \citenamefont {Russomanno},\ and\ \citenamefont {Santoro}}]{glen2020}%
  \BibitemOpen
  \bibfield  {author} {\bibinfo {author} {\bibfnamefont {G.~B.}\ \bibnamefont
  {Mbeng}}, \bibinfo {author} {\bibfnamefont {A.}~\bibnamefont {Russomanno}},\
  and\ \bibinfo {author} {\bibfnamefont {G.~E.}\ \bibnamefont {Santoro}},\
  }\bibfield  {title} {\bibinfo {title} {{The quantum Ising chain for
  beginners}},\ }\href {https://doi.org/10.21468/SciPostPhysLectNotes.82}
  {\bibfield  {journal} {\bibinfo  {journal} {SciPost Phys. Lect. Notes}\ ,\
  \bibinfo {pages} {82}} (\bibinfo {year} {2024})}\BibitemShut {NoStop}%
\bibitem [{\citenamefont {Solfanelli}\ \emph {et~al.}(2023)\citenamefont
  {Solfanelli}, \citenamefont {Ruffo}, \citenamefont {Succi},\ and\
  \citenamefont {Defenu}}]{solfanelli_jhep_2023}%
  \BibitemOpen
  \bibfield  {author} {\bibinfo {author} {\bibfnamefont {A.}~\bibnamefont
  {Solfanelli}}, \bibinfo {author} {\bibfnamefont {S.}~\bibnamefont {Ruffo}},
  \bibinfo {author} {\bibfnamefont {S.}~\bibnamefont {Succi}},\ and\ \bibinfo
  {author} {\bibfnamefont {N.}~\bibnamefont {Defenu}},\ }\bibfield  {title}
  {\bibinfo {title} {{Logarithmic, fractal and volume-law entanglement in a
  Kitaev chain with long-range hopping and pairing}},\ }\href
  {https://doi.org/10.1007/JHEP05(2023)066} {\bibfield  {journal} {\bibinfo
  {journal} {J. High Energy Phys.}\ }\textbf {\bibinfo {volume} {2023}}\bibinfo
   {number} { (5)},\ \bibinfo {pages} {1}}\BibitemShut {NoStop}%
\bibitem [{\citenamefont {Bohnet}\ \emph {et~al.}(2016)\citenamefont {Bohnet},
  \citenamefont {Sawyer}, \citenamefont {Britton}, \citenamefont {Wall},
  \citenamefont {Rey}, \citenamefont {Foss-Feig},\ and\ \citenamefont
  {Bollinger}}]{quantum_magnetism_ion_trap}%
  \BibitemOpen
\bibfield  {number} {  }\bibfield  {author} {\bibinfo {author} {\bibfnamefont
  {J.~G.}\ \bibnamefont {Bohnet}}, \bibinfo {author} {\bibfnamefont {B.~C.}\
  \bibnamefont {Sawyer}}, \bibinfo {author} {\bibfnamefont {J.~W.}\
  \bibnamefont {Britton}}, \bibinfo {author} {\bibfnamefont {M.~L.}\
  \bibnamefont {Wall}}, \bibinfo {author} {\bibfnamefont {A.~M.}\ \bibnamefont
  {Rey}}, \bibinfo {author} {\bibfnamefont {M.}~\bibnamefont {Foss-Feig}},\
  and\ \bibinfo {author} {\bibfnamefont {J.~J.}\ \bibnamefont {Bollinger}},\
  }\bibfield  {title} {\bibinfo {title} {Quantum spin dynamics and entanglement
  generation with hundreds of trapped ions},\ }\href
  {https://doi.org/10.1126/science.aad9958} {\bibfield  {journal} {\bibinfo
  {journal} {Science}\ }\textbf {\bibinfo {volume} {352}},\ \bibinfo {pages}
  {1297} (\bibinfo {year} {2016})},\ \Eprint
  {https://arxiv.org/abs/https://www.science.org/doi/pdf/10.1126/science.aad9958}
  {https://www.science.org/doi/pdf/10.1126/science.aad9958} \BibitemShut
  {NoStop}%
\bibitem [{\citenamefont {Hawaldar}\ \emph {et~al.}(2024)\citenamefont
  {Hawaldar}, \citenamefont {Shahi}, \citenamefont {Carter}, \citenamefont
  {Rey}, \citenamefont {Bollinger},\ and\ \citenamefont
  {Shankar}}]{athreya_bilayer_iontrapp_prx_2024}%
  \BibitemOpen
  \bibfield  {author} {\bibinfo {author} {\bibfnamefont {S.}~\bibnamefont
  {Hawaldar}}, \bibinfo {author} {\bibfnamefont {P.}~\bibnamefont {Shahi}},
  \bibinfo {author} {\bibfnamefont {A.~L.}\ \bibnamefont {Carter}}, \bibinfo
  {author} {\bibfnamefont {A.~M.}\ \bibnamefont {Rey}}, \bibinfo {author}
  {\bibfnamefont {J.~J.}\ \bibnamefont {Bollinger}},\ and\ \bibinfo {author}
  {\bibfnamefont {A.}~\bibnamefont {Shankar}},\ }\bibfield  {title} {\bibinfo
  {title} {Bilayer crystals of trapped ions for quantum information
  processing},\ }\href {https://doi.org/10.1103/PhysRevX.14.031030} {\bibfield
  {journal} {\bibinfo  {journal} {Phys. Rev. X}\ }\textbf {\bibinfo {volume}
  {14}},\ \bibinfo {pages} {031030} (\bibinfo {year} {2024})}\BibitemShut
  {NoStop}%
\bibitem [{\citenamefont {Jurcevic}\ \emph {et~al.}(2014)\citenamefont
  {Jurcevic}, \citenamefont {Lanyon}, \citenamefont {Hauke} \emph
  {et~al.}}]{Jurcevic2014Jul}%
  \BibitemOpen
  \bibfield  {author} {\bibinfo {author} {\bibfnamefont {P.}~\bibnamefont
  {Jurcevic}}, \bibinfo {author} {\bibfnamefont {B.~P.}\ \bibnamefont
  {Lanyon}}, \bibinfo {author} {\bibfnamefont {P.}~\bibnamefont {Hauke}}, \emph
  {et~al.},\ }\bibfield  {title} {\bibinfo {title} {{Quasiparticle engineering
  and entanglement propagation in a quantum many-body system}},\ }\href
  {https://doi.org/10.1038/nature13461} {\bibfield  {journal} {\bibinfo
  {journal} {Nature}\ }\textbf {\bibinfo {volume} {511}},\ \bibinfo {pages}
  {202} (\bibinfo {year} {2014})}\BibitemShut {NoStop}%
\bibitem [{\citenamefont {Pientka}\ \emph {et~al.}(2013)\citenamefont
  {Pientka}, \citenamefont {Glazman},\ and\ \citenamefont {von
  Oppen}}]{pientka_prb_2013}%
  \BibitemOpen
  \bibfield  {author} {\bibinfo {author} {\bibfnamefont {F.}~\bibnamefont
  {Pientka}}, \bibinfo {author} {\bibfnamefont {L.~I.}\ \bibnamefont
  {Glazman}},\ and\ \bibinfo {author} {\bibfnamefont {F.}~\bibnamefont {von
  Oppen}},\ }\bibfield  {title} {\bibinfo {title} {Topological superconducting
  phase in helical shiba chains},\ }\href
  {https://doi.org/10.1103/PhysRevB.88.155420} {\bibfield  {journal} {\bibinfo
  {journal} {Phys. Rev. B}\ }\textbf {\bibinfo {volume} {88}},\ \bibinfo
  {pages} {155420} (\bibinfo {year} {2013})}\BibitemShut {NoStop}%
\bibitem [{\citenamefont {Sinaga}\ \emph {et~al.}(2024)\citenamefont {Sinaga},
  \citenamefont {Adams}, \citenamefont {Hasdeo},\ and\ \citenamefont
  {Michels}}]{sinaga_prb_2024}%
  \BibitemOpen
  \bibfield  {author} {\bibinfo {author} {\bibfnamefont {E.~P.}\ \bibnamefont
  {Sinaga}}, \bibinfo {author} {\bibfnamefont {M.~P.}\ \bibnamefont {Adams}},
  \bibinfo {author} {\bibfnamefont {E.~H.}\ \bibnamefont {Hasdeo}},\ and\
  \bibinfo {author} {\bibfnamefont {A.}~\bibnamefont {Michels}},\ }\bibfield
  {title} {\bibinfo {title} {Neutron scattering signature of the
  dzyaloshinskii-moriya interaction in nanoparticles},\ }\href
  {https://doi.org/10.1103/PhysRevB.110.054404} {\bibfield  {journal} {\bibinfo
   {journal} {Phys. Rev. B}\ }\textbf {\bibinfo {volume} {110}},\ \bibinfo
  {pages} {054404} (\bibinfo {year} {2024})}\BibitemShut {NoStop}%
\bibitem [{\citenamefont {Vodola}\ \emph {et~al.}(2014)\citenamefont {Vodola},
  \citenamefont {Lepori}, \citenamefont {Ercolessi}, \citenamefont {Gorshkov},\
  and\ \citenamefont {Pupillo}}]{Vodola_prl_2014}%
  \BibitemOpen
  \bibfield  {author} {\bibinfo {author} {\bibfnamefont {D.}~\bibnamefont
  {Vodola}}, \bibinfo {author} {\bibfnamefont {L.}~\bibnamefont {Lepori}},
  \bibinfo {author} {\bibfnamefont {E.}~\bibnamefont {Ercolessi}}, \bibinfo
  {author} {\bibfnamefont {A.~V.}\ \bibnamefont {Gorshkov}},\ and\ \bibinfo
  {author} {\bibfnamefont {G.}~\bibnamefont {Pupillo}},\ }\bibfield  {title}
  {\bibinfo {title} {Kitaev chains with long-range pairing},\ }\href
  {https://doi.org/10.1103/PhysRevLett.113.156402} {\bibfield  {journal}
  {\bibinfo  {journal} {Phys. Rev. Lett.}\ }\textbf {\bibinfo {volume} {113}},\
  \bibinfo {pages} {156402} (\bibinfo {year} {2014})}\BibitemShut {NoStop}%
\bibitem [{\citenamefont {Vodola}\ \emph {et~al.}(2015)\citenamefont {Vodola},
  \citenamefont {Lepori}, \citenamefont {Ercolessi},\ and\ \citenamefont
  {Pupillo}}]{Vodola_njp_2015}%
  \BibitemOpen
  \bibfield  {author} {\bibinfo {author} {\bibfnamefont {D.}~\bibnamefont
  {Vodola}}, \bibinfo {author} {\bibfnamefont {L.}~\bibnamefont {Lepori}},
  \bibinfo {author} {\bibfnamefont {E.}~\bibnamefont {Ercolessi}},\ and\
  \bibinfo {author} {\bibfnamefont {G.}~\bibnamefont {Pupillo}},\ }\bibfield
  {title} {\bibinfo {title} {{Long-range Ising and Kitaev models: phases,
  correlations and edge modes}},\ }\href
  {https://doi.org/10.1088/1367-2630/18/1/015001} {\bibfield  {journal}
  {\bibinfo  {journal} {New J. Phys.}\ }\textbf {\bibinfo {volume} {18}},\
  \bibinfo {pages} {015001} (\bibinfo {year} {2015})}\BibitemShut {NoStop}%
\bibitem [{\citenamefont {Soltani}\ \emph {et~al.}(2019)\citenamefont
  {Soltani}, \citenamefont {Khastehdel~Fumani},\ and\ \citenamefont
  {Mahdavifar}}]{soltani_magn_2019}%
  \BibitemOpen
  \bibfield  {author} {\bibinfo {author} {\bibfnamefont {M.~R.}\ \bibnamefont
  {Soltani}}, \bibinfo {author} {\bibfnamefont {F.}~\bibnamefont
  {Khastehdel~Fumani}},\ and\ \bibinfo {author} {\bibfnamefont
  {S.}~\bibnamefont {Mahdavifar}},\ }\bibfield  {title} {\bibinfo {title}
  {{Ising in a transverse field with added transverse Dzyaloshinskii-Moriya
  interaction}},\ }\href {https://doi.org/10.1016/j.jmmm.2018.12.019}
  {\bibfield  {journal} {\bibinfo  {journal} {J. Magn. Magn. Mater.}\ }\textbf
  {\bibinfo {volume} {476}},\ \bibinfo {pages} {580} (\bibinfo {year}
  {2019})}\BibitemShut {NoStop}%
\bibitem [{\citenamefont {Luo}(2022{\natexlab{b}})}]{luo_prb_2022}%
  \BibitemOpen
  \bibfield  {author} {\bibinfo {author} {\bibfnamefont {Q.}~\bibnamefont
  {Luo}},\ }\bibfield  {title} {\bibinfo {title} {Analytical results for the
  unusual gr\"uneisen ratio in the quantum ising model with
  dzyaloshinskii-moriya interaction},\ }\href
  {https://doi.org/10.1103/PhysRevB.105.L060401} {\bibfield  {journal}
  {\bibinfo  {journal} {Phys. Rev. B}\ }\textbf {\bibinfo {volume} {105}},\
  \bibinfo {pages} {L060401} (\bibinfo {year}
  {2022}{\natexlab{b}})}\BibitemShut {NoStop}%
\bibitem [{\citenamefont {Groisman}\ \emph {et~al.}(2005)\citenamefont
  {Groisman}, \citenamefont {Popescu},\ and\ \citenamefont
  {Winter}}]{Winter05}%
  \BibitemOpen
  \bibfield  {author} {\bibinfo {author} {\bibfnamefont {B.}~\bibnamefont
  {Groisman}}, \bibinfo {author} {\bibfnamefont {S.}~\bibnamefont {Popescu}},\
  and\ \bibinfo {author} {\bibfnamefont {A.}~\bibnamefont {Winter}},\
  }\bibfield  {title} {\bibinfo {title} {Quantum, classical, and total amount
  of correlations in a quantum state},\ }\href
  {https://doi.org/10.1103/PhysRevA.72.032317} {\bibfield  {journal} {\bibinfo
  {journal} {Phys. Rev. A}\ }\textbf {\bibinfo {volume} {72}},\ \bibinfo
  {pages} {032317} (\bibinfo {year} {2005})}\BibitemShut {NoStop}%
\bibitem [{\citenamefont {Sadhukhan}\ and\ \citenamefont
  {Dziarmaga}(2021)}]{sadhukhan_arxiv_2021}%
  \BibitemOpen
  \bibfield  {author} {\bibinfo {author} {\bibfnamefont {D.}~\bibnamefont
  {Sadhukhan}}\ and\ \bibinfo {author} {\bibfnamefont {J.}~\bibnamefont
  {Dziarmaga}},\ }\href@noop {} {\bibinfo {title} {Is there a correlation
  length in a model with long-range interactions?}} (\bibinfo {year} {2021}),\
  \Eprint {https://arxiv.org/abs/2107.02508} {arXiv:2107.02508
  [cond-mat.str-el]} \BibitemShut {NoStop}%
\bibitem [{\citenamefont {Francica}\ and\ \citenamefont
  {Dell'Anna}(2022)}]{francica_prb_2022}%
  \BibitemOpen
  \bibfield  {author} {\bibinfo {author} {\bibfnamefont {G.}~\bibnamefont
  {Francica}}\ and\ \bibinfo {author} {\bibfnamefont {L.}~\bibnamefont
  {Dell'Anna}},\ }\bibfield  {title} {\bibinfo {title} {Correlations,
  long-range entanglement, and dynamics in long-range kitaev chains},\ }\href
  {https://doi.org/10.1103/PhysRevB.106.155126} {\bibfield  {journal} {\bibinfo
   {journal} {Phys. Rev. B}\ }\textbf {\bibinfo {volume} {106}},\ \bibinfo
  {pages} {155126} (\bibinfo {year} {2022})}\BibitemShut {NoStop}%
\bibitem [{\citenamefont {Amico}\ \emph {et~al.}(2008)\citenamefont {Amico},
  \citenamefont {Fazio}, \citenamefont {Osterloh},\ and\ \citenamefont
  {Vedral}}]{fazio_rev}%
  \BibitemOpen
  \bibfield  {author} {\bibinfo {author} {\bibfnamefont {L.}~\bibnamefont
  {Amico}}, \bibinfo {author} {\bibfnamefont {R.}~\bibnamefont {Fazio}},
  \bibinfo {author} {\bibfnamefont {A.}~\bibnamefont {Osterloh}},\ and\
  \bibinfo {author} {\bibfnamefont {V.}~\bibnamefont {Vedral}},\ }\bibfield
  {title} {\bibinfo {title} {Entanglement in many-body systems},\ }\href
  {https://doi.org/10.1103/RevModPhys.80.517} {\bibfield  {journal} {\bibinfo
  {journal} {Rev. Mod. Phys.}\ }\textbf {\bibinfo {volume} {80}},\ \bibinfo
  {pages} {517} (\bibinfo {year} {2008})}\BibitemShut {NoStop}%
\bibitem [{\citenamefont {Calabrese}\ and\ \citenamefont
  {Cardy}(2004)}]{Calabrese2004Jun}%
  \BibitemOpen
  \bibfield  {author} {\bibinfo {author} {\bibfnamefont {P.}~\bibnamefont
  {Calabrese}}\ and\ \bibinfo {author} {\bibfnamefont {J.}~\bibnamefont
  {Cardy}},\ }\bibfield  {title} {\bibinfo {title} {{Entanglement entropy and
  quantum field theory}},\ }\href
  {https://doi.org/10.1088/1742-5468/2004/06/P06002} {\bibfield  {journal}
  {\bibinfo  {journal} {J. Stat. Mech.: Theory Exp.}\ }\textbf {\bibinfo
  {volume} {2004}}\bibinfo  {number} { (06)},\ \bibinfo {pages}
  {P06002}}\BibitemShut {NoStop}%
\bibitem [{\citenamefont {Holzhey}\ \emph {et~al.}(1994)\citenamefont
  {Holzhey}, \citenamefont {Larsen},\ and\ \citenamefont
  {Wilczek}}]{holzhey_npb_1994}%
  \BibitemOpen
\bibfield  {number} {  }\bibfield  {author} {\bibinfo {author} {\bibfnamefont
  {C.}~\bibnamefont {Holzhey}}, \bibinfo {author} {\bibfnamefont
  {F.}~\bibnamefont {Larsen}},\ and\ \bibinfo {author} {\bibfnamefont
  {F.}~\bibnamefont {Wilczek}},\ }\bibfield  {title} {\bibinfo {title}
  {{Geometric and renormalized entropy in conformal field theory}},\ }\href
  {https://doi.org/10.1016/0550-3213(94)90402-2} {\bibfield  {journal}
  {\bibinfo  {journal} {Nucl. Phys. B}\ }\textbf {\bibinfo {volume} {424}},\
  \bibinfo {pages} {443} (\bibinfo {year} {1994})}\BibitemShut {NoStop}%
\bibitem [{\citenamefont {Yang}\ \emph {et~al.}(2024)\citenamefont {Yang},
  \citenamefont {Lin},\ and\ \citenamefont {Yu}}]{yang_arxiv_2024}%
  \BibitemOpen
  \bibfield  {author} {\bibinfo {author} {\bibfnamefont {S.}~\bibnamefont
  {Yang}}, \bibinfo {author} {\bibfnamefont {H.-Q.}\ \bibnamefont {Lin}},\ and\
  \bibinfo {author} {\bibfnamefont {X.-J.}\ \bibnamefont {Yu}},\ }\href
  {https://arxiv.org/abs/2406.01974} {\bibinfo {title} {Gifts from long-range
  interaction: Emergent gapless topological behaviors in quantum spin chain}}
  (\bibinfo {year} {2024}),\ \Eprint {https://arxiv.org/abs/2406.01974}
  {arXiv:2406.01974 [cond-mat.str-el]} \BibitemShut {NoStop}%
\bibitem [{\citenamefont {Chakraborty}\ and\ \citenamefont
  {Angelinos}(2024)}]{chakraborty_arxiv_2024}%
  \BibitemOpen
  \bibfield  {author} {\bibinfo {author} {\bibfnamefont {D.}~\bibnamefont
  {Chakraborty}}\ and\ \bibinfo {author} {\bibfnamefont {N.}~\bibnamefont
  {Angelinos}},\ }\href {https://arxiv.org/abs/2302.06743} {\bibinfo {title}
  {Entanglement entropy in ground states of long-range fermionic systems}}
  (\bibinfo {year} {2024}),\ \Eprint {https://arxiv.org/abs/2302.06743}
  {arXiv:2302.06743 [cond-mat.str-el]} \BibitemShut {NoStop}%
\bibitem [{\citenamefont {Lau}\ \emph {et~al.}(2022)\citenamefont {Lau},
  \citenamefont {Haug}, \citenamefont {Kwek},\ and\ \citenamefont
  {Bharti}}]{kishor_bharti_truncated_scipost}%
  \BibitemOpen
  \bibfield  {author} {\bibinfo {author} {\bibfnamefont {J.~W.~Z.}\
  \bibnamefont {Lau}}, \bibinfo {author} {\bibfnamefont {T.}~\bibnamefont
  {Haug}}, \bibinfo {author} {\bibfnamefont {L.~C.}\ \bibnamefont {Kwek}},\
  and\ \bibinfo {author} {\bibfnamefont {K.}~\bibnamefont {Bharti}},\
  }\bibfield  {title} {\bibinfo {title} {{NISQ Algorithm for Hamiltonian
  simulation via truncated Taylor series}},\ }\href
  {https://doi.org/10.21468/SciPostPhys.12.4.122} {\bibfield  {journal}
  {\bibinfo  {journal} {SciPost Phys.}\ }\textbf {\bibinfo {volume} {12}},\
  \bibinfo {pages} {122} (\bibinfo {year} {2022})}\BibitemShut {NoStop}%
\bibitem [{\citenamefont {Ippoliti}\ \emph {et~al.}(2021)\citenamefont
  {Ippoliti}, \citenamefont {Kechedzhi}, \citenamefont {Moessner},
  \citenamefont {Sondhi},\ and\ \citenamefont
  {Khemani}}]{sycamore_dtc_prx_quantum}%
  \BibitemOpen
  \bibfield  {author} {\bibinfo {author} {\bibfnamefont {M.}~\bibnamefont
  {Ippoliti}}, \bibinfo {author} {\bibfnamefont {K.}~\bibnamefont {Kechedzhi}},
  \bibinfo {author} {\bibfnamefont {R.}~\bibnamefont {Moessner}}, \bibinfo
  {author} {\bibfnamefont {S.}~\bibnamefont {Sondhi}},\ and\ \bibinfo {author}
  {\bibfnamefont {V.}~\bibnamefont {Khemani}},\ }\bibfield  {title} {\bibinfo
  {title} {Many-body physics in the nisq era: Quantum programming a discrete
  time crystal},\ }\href {https://doi.org/10.1103/PRXQuantum.2.030346}
  {\bibfield  {journal} {\bibinfo  {journal} {PRX Quantum}\ }\textbf {\bibinfo
  {volume} {2}},\ \bibinfo {pages} {030346} (\bibinfo {year}
  {2021})}\BibitemShut {NoStop}%
\bibitem [{\citenamefont {Cao}\ \emph {et~al.}(2024)\citenamefont {Cao},
  \citenamefont {Hu}, \citenamefont {Tong},\ and\ \citenamefont
  {Yang}}]{cao_prb_2024}%
  \BibitemOpen
  \bibfield  {author} {\bibinfo {author} {\bibfnamefont {K.}~\bibnamefont
  {Cao}}, \bibinfo {author} {\bibfnamefont {Y.}~\bibnamefont {Hu}}, \bibinfo
  {author} {\bibfnamefont {P.}~\bibnamefont {Tong}},\ and\ \bibinfo {author}
  {\bibfnamefont {G.}~\bibnamefont {Yang}},\ }\bibfield  {title} {\bibinfo
  {title} {Dynamical relaxation behavior of an extended xy chain with a gapless
  phase following a quantum quench},\ }\href
  {https://doi.org/10.1103/PhysRevB.109.024303} {\bibfield  {journal} {\bibinfo
   {journal} {Phys. Rev. B}\ }\textbf {\bibinfo {volume} {109}},\ \bibinfo
  {pages} {024303} (\bibinfo {year} {2024})}\BibitemShut {NoStop}%
\bibitem [{\citenamefont {Fagotti}\ and\ \citenamefont
  {Calabrese}(2008)}]{XY_entropy_growth_calabrese_2008}%
  \BibitemOpen
  \bibfield  {author} {\bibinfo {author} {\bibfnamefont {M.}~\bibnamefont
  {Fagotti}}\ and\ \bibinfo {author} {\bibfnamefont {P.}~\bibnamefont
  {Calabrese}},\ }\bibfield  {title} {\bibinfo {title} {Evolution of
  entanglement entropy following a quantum quench: Analytic results for the
  $xy$ chain in a transverse magnetic field},\ }\href
  {https://doi.org/10.1103/PhysRevA.78.010306} {\bibfield  {journal} {\bibinfo
  {journal} {Phys. Rev. A}\ }\textbf {\bibinfo {volume} {78}},\ \bibinfo
  {pages} {010306} (\bibinfo {year} {2008})}\BibitemShut {NoStop}%
\bibitem [{\citenamefont {Schachenmayer}\ \emph {et~al.}(2013)\citenamefont
  {Schachenmayer}, \citenamefont {Lanyon}, \citenamefont {Roos},\ and\
  \citenamefont {Daley}}]{Schachenmayer_entropy_growth_2013}%
  \BibitemOpen
  \bibfield  {author} {\bibinfo {author} {\bibfnamefont {J.}~\bibnamefont
  {Schachenmayer}}, \bibinfo {author} {\bibfnamefont {B.~P.}\ \bibnamefont
  {Lanyon}}, \bibinfo {author} {\bibfnamefont {C.~F.}\ \bibnamefont {Roos}},\
  and\ \bibinfo {author} {\bibfnamefont {A.~J.}\ \bibnamefont {Daley}},\
  }\bibfield  {title} {\bibinfo {title} {Entanglement growth in quench dynamics
  with variable range interactions},\ }\href
  {https://doi.org/10.1103/PhysRevX.3.031015} {\bibfield  {journal} {\bibinfo
  {journal} {Phys. Rev. X}\ }\textbf {\bibinfo {volume} {3}},\ \bibinfo {pages}
  {031015} (\bibinfo {year} {2013})}\BibitemShut {NoStop}%
\bibitem [{\citenamefont {Van~Regemortel}\ \emph {et~al.}(2016)\citenamefont
  {Van~Regemortel}, \citenamefont {Sels},\ and\ \citenamefont
  {Wouters}}]{regemortel_pra_2016}%
  \BibitemOpen
  \bibfield  {author} {\bibinfo {author} {\bibfnamefont {M.}~\bibnamefont
  {Van~Regemortel}}, \bibinfo {author} {\bibfnamefont {D.}~\bibnamefont
  {Sels}},\ and\ \bibinfo {author} {\bibfnamefont {M.}~\bibnamefont
  {Wouters}},\ }\bibfield  {title} {\bibinfo {title} {Information propagation
  and equilibration in long-range kitaev chains},\ }\href
  {https://doi.org/10.1103/PhysRevA.93.032311} {\bibfield  {journal} {\bibinfo
  {journal} {Phys. Rev. A}\ }\textbf {\bibinfo {volume} {93}},\ \bibinfo
  {pages} {032311} (\bibinfo {year} {2016})}\BibitemShut {NoStop}%
\bibitem [{\citenamefont {Buyskikh}\ \emph {et~al.}(2016)\citenamefont
  {Buyskikh}, \citenamefont {Fagotti}, \citenamefont {Schachenmayer},
  \citenamefont {Essler},\ and\ \citenamefont {Daley}}]{buyskikh_pra_2016}%
  \BibitemOpen
  \bibfield  {author} {\bibinfo {author} {\bibfnamefont {A.~S.}\ \bibnamefont
  {Buyskikh}}, \bibinfo {author} {\bibfnamefont {M.}~\bibnamefont {Fagotti}},
  \bibinfo {author} {\bibfnamefont {J.}~\bibnamefont {Schachenmayer}}, \bibinfo
  {author} {\bibfnamefont {F.}~\bibnamefont {Essler}},\ and\ \bibinfo {author}
  {\bibfnamefont {A.~J.}\ \bibnamefont {Daley}},\ }\bibfield  {title} {\bibinfo
  {title} {Entanglement growth and correlation spreading with variable-range
  interactions in spin and fermionic tunneling models},\ }\href
  {https://doi.org/10.1103/PhysRevA.93.053620} {\bibfield  {journal} {\bibinfo
  {journal} {Phys. Rev. A}\ }\textbf {\bibinfo {volume} {93}},\ \bibinfo
  {pages} {053620} (\bibinfo {year} {2016})}\BibitemShut {NoStop}%
\bibitem [{\citenamefont {Maity}\ \emph {et~al.}(2019)\citenamefont {Maity},
  \citenamefont {Bhattacharya},\ and\ \citenamefont
  {Dutta}}]{Maity2019_review}%
  \BibitemOpen
  \bibfield  {author} {\bibinfo {author} {\bibfnamefont {S.}~\bibnamefont
  {Maity}}, \bibinfo {author} {\bibfnamefont {U.}~\bibnamefont
  {Bhattacharya}},\ and\ \bibinfo {author} {\bibfnamefont {A.}~\bibnamefont
  {Dutta}},\ }\bibfield  {title} {\bibinfo {title} {{One-dimensional quantum
  many body systems with long-range interactions}},\ }\href
  {https://doi.org/10.1088/1751-8121/ab5634} {\bibfield  {journal} {\bibinfo
  {journal} {J. Phys. A: Math. Theor.}\ }\textbf {\bibinfo {volume} {53}},\
  \bibinfo {pages} {013001} (\bibinfo {year} {2019})}\BibitemShut {NoStop}%
\bibitem [{\citenamefont {Calabrese}\ and\ \citenamefont
  {Cardy}(2005)}]{evolution_entropy_calabrese}%
  \BibitemOpen
  \bibfield  {author} {\bibinfo {author} {\bibfnamefont {P.}~\bibnamefont
  {Calabrese}}\ and\ \bibinfo {author} {\bibfnamefont {J.}~\bibnamefont
  {Cardy}},\ }\bibfield  {title} {\bibinfo {title} {Evolution of entanglement
  entropy in one-dimensional systems},\ }\href
  {https://doi.org/10.1088/1742-5468/2005/04/P04010} {\bibfield  {journal}
  {\bibinfo  {journal} {Journal of Statistical Mechanics: Theory and
  Experiment}\ }\textbf {\bibinfo {volume} {2005}},\ \bibinfo {pages} {P04010}
  (\bibinfo {year} {2005})}\BibitemShut {NoStop}%
\bibitem [{\citenamefont {Wimmer}(2012)}]{Wimmer2012Aug}%
  \BibitemOpen
  \bibfield  {author} {\bibinfo {author} {\bibfnamefont {M.}~\bibnamefont
  {Wimmer}},\ }\bibfield  {title} {\bibinfo {title} {{Algorithm 923: Efficient
  Numerical Computation of the Pfaffian for Dense and Banded Skew-Symmetric
  Matrices}},\ }\href {https://doi.org/10.1145/2331130.2331138} {\bibfield
  {journal} {\bibinfo  {journal} {ACM Trans. Math. Softw.}\ }\textbf {\bibinfo
  {volume} {38}},\ \bibinfo {pages} {1} (\bibinfo {year} {2012})}\BibitemShut
  {NoStop}%
\end{thebibliography}%

\appendix

 \section{Correlation functions}%Calculation of correlation functions using  Pfaffian}
\label{pfaff_corre}

%In order to calculate the mutual information, we first need 
The two-point correlation functions are the fundamental constituents of our study. In the context of the ground state or the thermal state, the two-site correlation function between sites separated by a distance of $R$ can usually be expressed as a determinant of an $R \times R$ Toeplitz matrix \cite{lieb1961, barouch_pra_1970_1, barouch_pra_1970_2}.
However, as the Hamiltonian includes asymmetric interaction, the bipartite reduced state obtained from the \(N\)-party ground  state after tracing out \(N-2\) parties includes the correlators like $\mathcal{C}^{xy} \equiv \langle \sigma^x_i \sigma^y_{i+R}\rangle = \text{tr} (\sigma^x \otimes \sigma^y \rho_R) $ and $\mathcal{C}^{yx}$ along with  \(\mathcal{C}^{ii}\) (\(i=x,y,z\)).  Thus for the corresponding density matrix of the two distant parties of the spin chain and the time-evolved state, we need to deal directly with the Pfaffians in order to evaluate the correlations as a function of time as well as distance since there is a formation of chiral phase due to DM interaction. 

Using the Pfaffian formalism, we first write the spin correlation functions as 
%Our technique is quite different than the traditional technique where   
%we first need 
%need the time evolved two-party reduced density matrix between distant sites \(i\) and \(i+R\). The corresponding time-evolved state can be written as $\rho(t) = e^{-iH^+t} \rho e^{iH^+t} $ with which expectation value of any operator in Fourier basis can be calculated.  

%We calculate the correlation functions of reduced state density matrix $\rho_R$ using  Pfaffian formalism. 
%\cite{pfaff_amico_2004}. 

\begin{widetext}
    \begin{eqnarray}
& \left .
\begin{array}{clllllcll} 
\mathcal{C}^{lm}_{i,i+R} = \langle \sigma_i^l \sigma_{i+R}^m \rangle= {c(l,m)} {\mathrm{pf}}\; \left | \right . I^{lm}_{1,2} & \dots  & I^{lm}_{1,R-1}  &J^{lm}_{1} & F^{lm}_{1} 
& G^{lm}_{1,2} &\phantom{c} .  &  \dots & G^{lm}_{1,r}  \\
           & \dots  & \dots &\dots &  \dots & \dots  & \phantom{c} . & \dots & \dots \\
           &        & I^{lm}_{R-2,R-1}  & J^{lm}_{R-2}         & F^{lm}_{R-2}         & G^{lm}_{R-2,2} & \phantom{c} .  & \dots & G^{lm}_{R-2,R}         \\
           &        &   & J^{lm}_{R-1}          & F^{lm}_{R-1} &  G^{lm}_{R-1,2} & \phantom{c} .& \dots &  G^{lm}_{R-1,R}  \\
           &        &   &           & E^{lm} &  D^{lm}_{2} & \phantom{c} . & \dots &  D^{lm}_{R}  \\
           &        &   &           &             & K^{lm}_{2} & \phantom{c} .&\dots &  K^{lm}_{R} \\  
           &        &   &           &             &                       & H^{lm}_{2,3}  &\dots &  H^{lm}_{2,R} \\  
           &        &   &           &             &                     & &\dots & \dots         \\  
           &        &   &           &             &                  &   &  & H^{lm}_{R-1,R}  

\end{array}
\right |,&
\label{pfaffian}
\end{eqnarray}

\end{widetext}

where $c(x,x)=c(y,y)=(-1)^{R(R+1)/2}$, 
\begin{eqnarray}
I^{xx}_{\mu,\nu}&&=  \langle A_{l+\mu}(t) A_{l+\nu}(t)\rangle, \nonumber \\
J^{xx}_{\mu},  &&=I^{xx}_{\mu,R}, \nonumber \\ 
H^{xx}_{\mu,\nu}&&= \langle B_{l+\mu-1}(t) B_{l+\nu-1}(t)\rangle, \nonumber \\
K^{xx}_{\nu}&&= H^{xx}_{1,\nu}, \\ \nonumber 
G^{xx}_{\mu,\nu}&&= \langle A_{l+\mu}(t) B_{l+\nu-1}(t)\rangle, \label{xx}\\ 
F^{xx}_{\mu}&&=G^{xx}_{\mu,1},  \nonumber \\
E^{xx}&&=G^{xx}_{R,1},  \nonumber \\
D^{xx}_{\nu}&&=G^{xx}_{R,\nu},  \nonumber 
\end{eqnarray}
\begin{eqnarray}
 I^{yy}_{\mu,\nu}&&=  \langle A_{l+\mu-1}(t) A_{l+\nu-1}(t)\rangle, \nonumber \\
J^{yy}_{\mu}  &&=I^{yy}_{\mu,R}, \nonumber \\ 
H^{yy}_{\mu,\nu}&&= \langle B_{l+\mu}(t) B_{l+\nu}(t)\rangle, \nonumber \\
K^{yy}_{\nu}&&= H^{yy}_{1,\nu}, \\ \nonumber 
G^{yy}_{\mu,\nu}&&= \langle A_{l+\mu-1}(t) B_{l+\nu}(t)\rangle, \label{yy} \\ 
F^{yy}_{\mu}&&=G^{yy}_{\mu,1},  \nonumber \\
E^{yy}&&=G^{yy}_{R,1},  \nonumber \\
D^{yy}_{\nu}&&=G^{yy}_{R,\nu},  \nonumber 
\end{eqnarray}
with $s(x,y)=s(y,x)= -i (-1)^{R(R-1)/2}$, we have
\begin{eqnarray}
I^{xy}_{\mu,\nu}&&=  \langle A_{l+\mu}(t) A_{l+\nu}(t)\rangle, \nonumber \\
G^{xy}_{\mu,\nu}&&= \langle A_{l+\mu}(t) B_{l+\nu}(t)\rangle, \nonumber  \\
J^{xy}_{\mu}  &&= G^{xy}_{\mu,0}, \label{xy} \\ 
F^{xy}_{\mu}&&=G^{xy}_{\mu,1},  \nonumber \\
H^{xy}_{\mu,\nu}&&= \langle B_{l+\mu}(t) B_{l+\nu}(t)\rangle, \nonumber \\
E^{xy}&&=H^{xy}_{0,1},  \nonumber \\
D^{xy}_{\nu}&&= H^{xy}_{0,\nu}, \nonumber \\  
K^{xy}_{\nu}&&=H^{xy}_{1,\nu},  \nonumber 
\end{eqnarray}
and
\begin{eqnarray}
I^{yx}_{\mu,\nu}&&=  \langle A_{l+\mu-1}(t) A_{l+\nu-1}(t)\rangle, \nonumber \\
G^{yx}_{\mu,\nu}&&= \langle A_{l+\mu-1}(t) B_{l+\nu-1}(t)\rangle, \nonumber  \\
J^{yx}_{\mu}  &&= I^{yx}_{\mu,R}, \label{yx} \\ 
F^{yx}_{\mu}&&=I^{yx}_{\mu,R+1},  \nonumber \\
E^{yx}&&=I^{yx}_{r,r+1},  \nonumber \\
D^{yx}_{\nu}&&= G^{yx}_{R,\nu},  \nonumber \\
K^{yx}_{\nu}&&=G^{yx}_{R+1,\nu}, \nonumber  \\ 
H^{yx}_{\mu,\nu}&&= \langle B_{l+\mu-1}(t) B_{l+\nu-1}(t)\rangle. \nonumber 
\end{eqnarray}
Each of the elements in the Pfaffian can be constructed from the expectation values of one of the operators, $A_l A_{l+R}$, $B_l B_{l+R}$ and $A_l B_{l+R}$ with
%\begin{equation}
${A}_{i}=c_{i}^{\dagger}+c_{i}, \quad {B}_{i}=c_{i}^{\dagger}-c_{i}.$
%\end{equation}
%

\begin{widetext}

When recasted in the same Fourier basis as taken while diagonalization, the operators for the $k^\text{th}$ momentum have the matrix form
\begin{align}
 (A_l A_{l+R})_k =  \left[\begin{array}{cccc}
\cos(kR) & -i\sin(kR) & 0 & 0 \\
-i\sin(kR) & \cos(kR) & 0 & 0 \\
0&0&\cos(kR)-i\sin(kR)&0\\
0&0&0&\cos(kR)+i\sin(kR)
\end{array}\right],
\end{align}
\begin{align}
 (B_l B_{l+R})_k =   \left[\begin{array}{cccc}
-\cos(kR) & -i\sin(kR) & 0 & 0 \\
-i\sin(kR) & -\cos(kR) & 0 & 0 \\
0&0&-\cos(kR)+i\sin(kR)&0\\
0&0&0&-\cos(kR)-i\sin(kR)
\end{array}\right],
\end{align}

\begin{align}
 (A_l B_{l+R})_k =   \left[\begin{array}{cccc}
\cos(kR) & i\sin(kR) & 0 & 0 \\
-i\sin(kR) & -\cos(kR) & 0 & 0 \\
0&0&0&0\\
0&0&0&0
\end{array}\right].
\end{align} 
The single-site transverse magnetization in the same Fourier basis is given by
\begin{eqnarray}
    \sigma_z^k = \left[\begin{array}{cccc}
-1 & 0 & 0 & 0 \\
0 & 1 & 0 & 0 \\
0&0&0&0\\
0&0&0&0
\end{array}\right]. 
\end{eqnarray}

\end{widetext}
% Poorly written. 
% Describe the formalism step by step.
% }
The corresponding expectation value of the operators for each value of $k$, say $O^k$, with respect to the time-evolved  state is computed as $  \langle O \rangle = \sum_{k=1}^{{N/2}}\text{Tr}(\rho_\beta^k(t) O^k).$
%In order to define the,  we require the correlations corresponding to the 
Once we find all the single-site magnetizations $m_j^\alpha ~\forall~ \alpha=\{x,y,z\}$ at $j=i$ and $j=i+R$ and all possible two-site correlation functions $\mathcal{C}_{i,i+R}^{l,m} ~\forall~ l,m=\{x,y,z\}$  from the Pfaffians described above, we can construct the two-site reduced density matrix between sites \(i\) and \(I+R\) as we know any two party density matrix for the given Hamiltonian is given as

\begin{equation}
    \rho_{ij}=\frac{1}{4} \left (\mathbb{I}_4 + m^z({\sigma^z_j}+{\sigma^z_i}) + \sum_{k,l} \mathcal{C}^{kl}_{ij} \sigma^{k}_i \otimes \sigma^{l}_j\right),
\end{equation}
where \(k,l\in\{x,y\}\).

%$\rho_R \equiv \rho_{i,i+R}$. 

\end{document}